\documentclass{article}

\usepackage{amsfonts,amssymb,amsmath,a4wide,paralist}
\usepackage{epsfig,rotating,framed,enumerate,listings}
\usepackage{color}
\usepackage{paralist}
\usepackage{marginnote}

\newcommand{\OPEN}{{\it open}}

\newcommand{\PL}{{\it plt}}
\newcommand{\PLS}{{\it plts}}
\newcommand{\FI}{{\it first}}
\newcommand{\LA}{{\it last}}

\newcommand{\FRONT}{{\it front}}
\newcommand{\NN}{{\mathbb N}}

\newcommand{\bigo}{\ensuremath{\mathcal{O}}}

\newcommand{\qed}{\null\hfill $\square$\par\medskip}

\newtheorem{theorem}{Theorem}[section]
\newtheorem{example}[theorem]{Example}

\newtheorem{proposition}[theorem]{Proposition}
\newtheorem{corollary}[theorem]{Corollary}

\newcommand{\dpw} {\text{d-pw}}
\newenvironment{proof}{\noindent{\bf Proof~}}{\null\hfill $\Box$\par\medskip}

\newtheorem{remark}[theorem]{Remark}
\newenvironment{desctight}
  {\begin{list}{}{
\setlength\labelwidth{-5pt}
        \setlength{\itemsep}{0.5pt}
        \setlength{\parsep}{0pt}
        \setlength\itemindent{-\leftmargin}
        
}}
    {\end{list}}

\definecolor{light-gray}{gray}{0.5}
\newcommand{\gr}  {\color{light-gray}}

\newcommand{\dvsns} {\text{d-vsn}}

\newcommand{\IN}{\mathbb{N}}
\newcommand{\fpt} {\mbox{FPT}}
\newcommand{\xp} {\mbox{XP}}
\newcommand{\w} {\mbox{W}}

\begin{document}

\title{Integer Programming Models and Parameterized Algorithms for Controlling Palletizers\thanks{Short versions of this paper appeared in Proceedings of the 
{\em International Conference on Modelling, Computation and Optimization in Information Systems and Management Sciences} 
(MCO 2015) \cite{GRW15a}, {\em International Conference on Combinatorial Optimization and Applications} (COCOA 2015)  \cite{GRW15b},
and {\em International Conference on Operations Research} (OR 2013) \cite{GRW14a}, (OR 2014) \cite{GRW16}, (OR 2015) \cite{GRW16a}.}}

\author{Frank Gurski\thanks{University of  D\"usseldorf,
Institute of Computer Science, Algorithmics for Hard Problems Group, 40225 D\"usseldorf, Germany,
{\tt frank.gurski@hhu.de}}
\and 
Jochen Rethmann\thanks{Niederrhein University of Applied Sciences,
Faculty of Electrical Engineering and Computer Science, 47805 Krefeld,
Germany, 
{\tt jochen.rethmann@hs-niederrhein.de}}
\and
Egon Wanke\thanks{University of  D\"usseldorf,
Institute of Computer Science,  Algorithms and Data Structures Group, 40225 D\"usseldorf, Germany,
{\tt e.wanke@hhu.de}
}}

\maketitle

\begin{abstract}
We study the combinatorial {\sc FIFO Stack-Up} problem, where
bins have to be stacked-up from conveyor belts onto pallets.
This is done by palletizers or palletizing robots. Given $k$
sequences of labeled bins and a positive integer $p$,
the goal is to stack-up the bins by iteratively removing the first bin
of one of the $k$ sequences and put it onto a pallet located at one of
$p$ stack-up places. Each of these pallets has to contain bins of only
one label, bins of different labels have to be placed on different pallets.
After all bins of one label have been removed from the given sequences,
the corresponding stack-up place becomes available for a pallet of bins
of another label. All bins have the same size.
The {\sc FIFO Stack-Up} problem asks whether there is some processing of the
sequences of bins such that at most $p$ stack-up places are used.

In this paper we strengthen the hardness of the {\sc FIFO Stack-Up}  
shown in \cite{GRW16b} by considering practical cases and the distribution
of the pallets onto the sequences.
We  introduce a digraph model for this problem, the so called
decision graph, which allows us to 
give a breadth first search solution
of running time $\bigo(n^2 \cdot (m+2)^k)$, 
where $m$ represents the number of 
pallets and $n$ denotes the total number of bins in all sequences. 
Further we apply methods to solve hard problems to the  {\sc FIFO Stack-Up} problem. 
Therefor we consider restricted versions of the problem,
two integer programming models,
exponential time algorithms, 
parameterized algorithms, and
approximation algorithms.
In order to evaluate our algorithms,
we introduce a method to generate random, but realistic instances for
the {\sc FIFO Stack-Up} problem.
Our experimental study of running times shows that the breadth first search solution
on the decision graph combined with a cutting technique can be used to solve practical
instances on several thousands of bins of the {\sc FIFO Stack-Up}  problem. 
Further we analyze our two integer programming approaches implemented in CPLEX 
and GLPK.
As expected CPLEX can solve the instances much faster than GLPK and 
our pallet solution approach is much better than the bin solution approach.

\end{abstract}

%%%%%%%%%%%%%%%%%%%%%%%%%%%%%%%%%%%%%%%%%%%%%%%%%%%%%%%%%%%%%%%%%%%%%%%%%%
\section{Introduction}
%%%%%%%%%%%%%%%%%%%%%%%%%%%%%%%%%%%%%%%%%%%%%%%%%%%%%%%%%%%%%%%%%%%%%%%%%%

This paper is about the combinatorial problem of stacking up bins from
conveyor belts onto pallets. This problem originally appears in
{\em stack-up systems} that play an important role in delivery industry
and warehouses. Stack-up systems are often the back end of
{\em order-picking systems}. A detailed description of the applied background of
such systems is given in \cite{Kos94,RW97}. Logistic experiences over 30
years lead to high flexible conveyor-based stack-up systems in delivery
industry. We do not intend to modify the architecture of
existing systems, but try to develop efficient algorithms to control them. 
% \cite{Yam09}

The bins that have to be stacked-up onto pallets reach the
stack-up system on a main conveyor belt. At the end of the conveyor belt the bins
enter the palletizing system. Here
the bins are picked-up by stacker cranes or robotic arms and moved
onto pallets, which are located at {\em stack-up places}. 
Often vacuum grippers are used to pick-up the bins. This picking process can
be performed in different ways depending on the architecture
of the palletizing system (single-line or multi-line palletizers).
 Full pallets are carried away by automated guided vehicles, or
by another conveyor system, while new empty pallets are placed
at free stack-up places. 

The developers and producers of robotic palletizers distinguish between single-line
and multi-line palletizing systems. Each of these systems has its advantages and 
disadvantages.

In {\em single-line palletizing systems} there is only one conveyor belt from
which the bins are picked-up. Several robotic arms or stacker cranes are
placed around the end of the conveyor. We model such systems by a random
access storage which is automatically replenished with bins from the main conveyor, 
see Figure \ref{F10a}. 
The area from which the bins can be picked-up is called the storage area. 
The storage area is the last part of the conveyor belt where the stacker
cranes or robotic arms reach the bins.

\begin{figure}[h]
\centerline{\epsfxsize=90mm \epsfbox{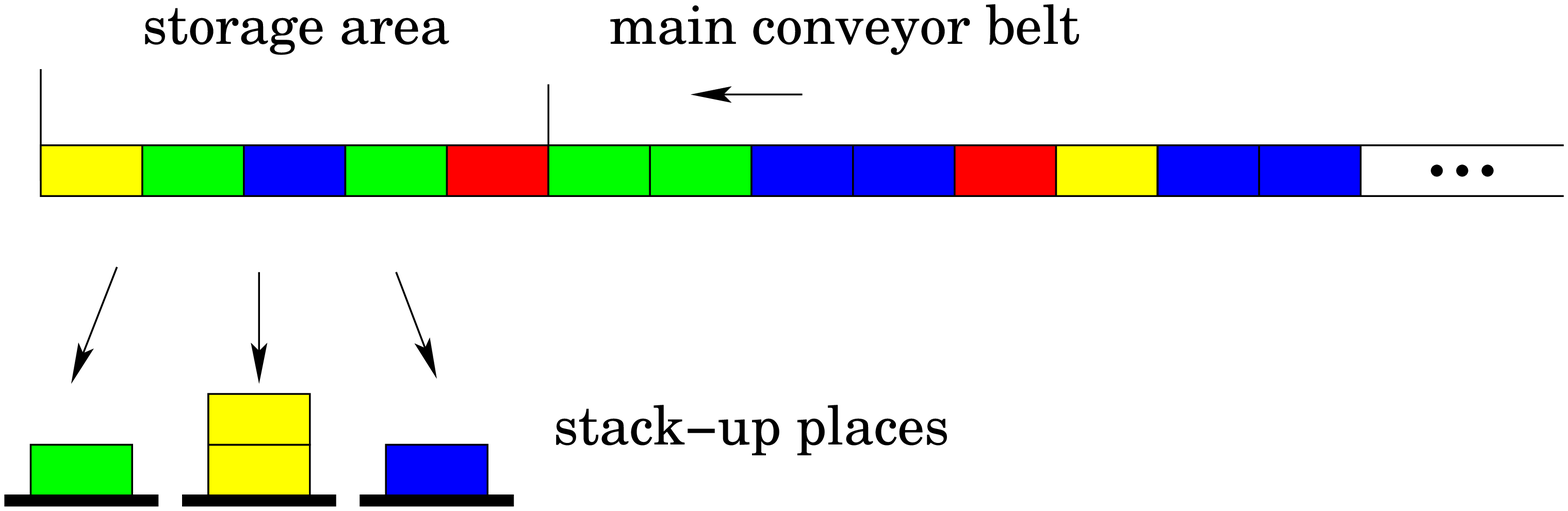}}
\caption{The single-line stack-up system using a random access
storage of storage capacity 5. The colors represent the pallet labels. 
Bins with different colors have to be placed on different pallets, bins 
with the same color have to be placed on the same pallet.}
\label{F10a}
\end{figure}

In {\em multi-line palletizing systems} there are several buffer conveyors from which
the bins are picked-up. The robotic arms or stacker cranes are placed at
the end of these conveyors. Here, the bins from the main conveyor of the order-picking
system first have to be distributed to the multiple infeed lines to enable
parallel processing. Such a distribution can be done by some cyclic storage conveyor,
see Figure \ref{F00}. 
From the cyclic storage conveyor the bins are pushed out to the buffer conveyors.
A stack-up system using a cyclic storage conveyor is, for example, located at
Bertelsmann Distribution GmbH in G\"utersloh, Germany. On certain days,
several thousands of bins are stacked-up using a cyclic storage conveyor with
a capacity of approximately 60 bins and 24 stack-up places, while up to
32 bins are destined for a pallet.
This palletizing system has originally initiated our research.

\begin{figure}[hbtp]
\centering
\parbox[b]{68mm}{
\centerline{\epsfxsize=64mm \epsfbox{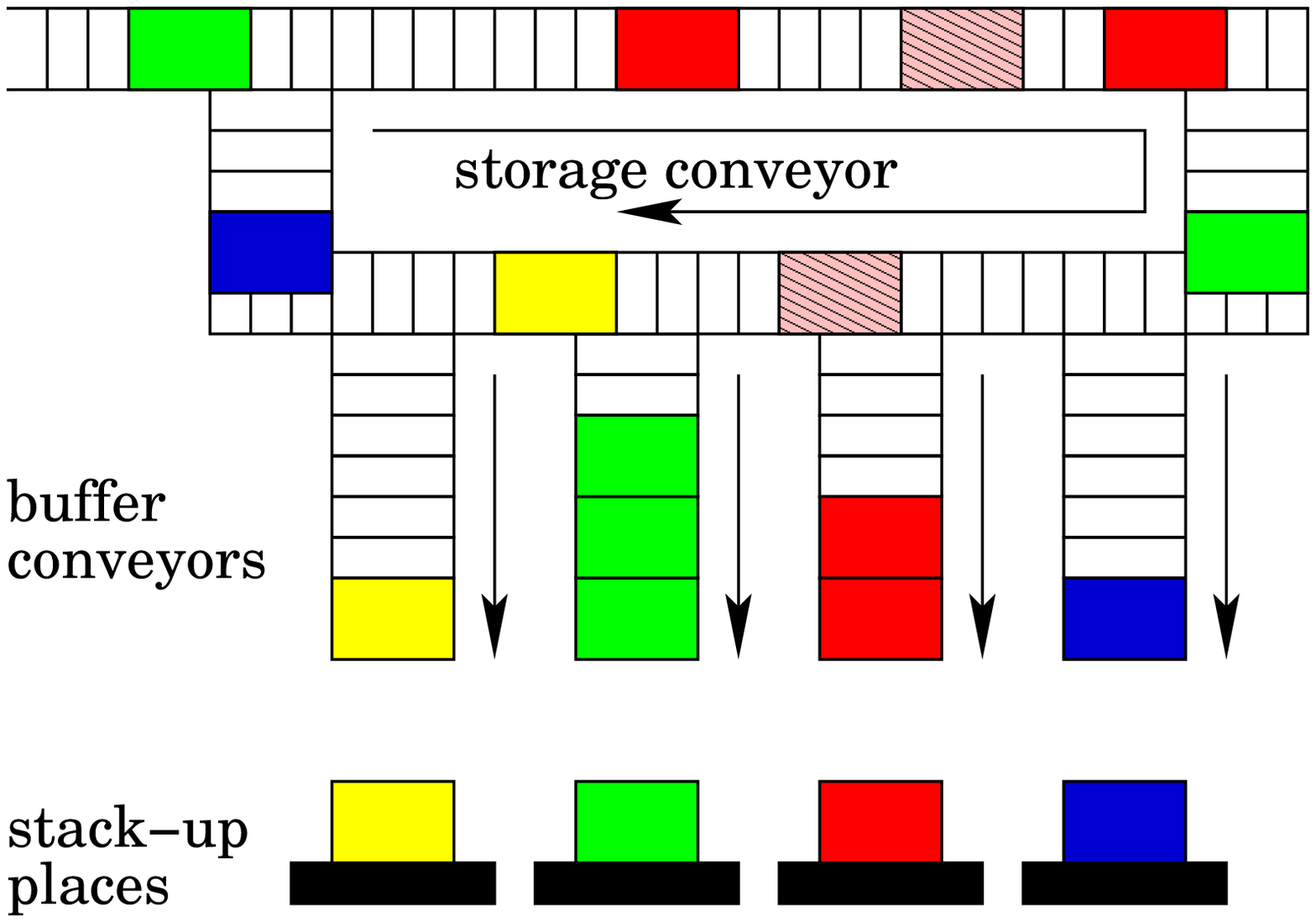}}
\caption{A multi-line stack-up system with a pre-placed cyclic storage conveyor.}
\label{F00}}
\hspace{1.0cm}
\parbox[b]{68mm}{
\centerline{\epsfxsize=64mm \epsfbox{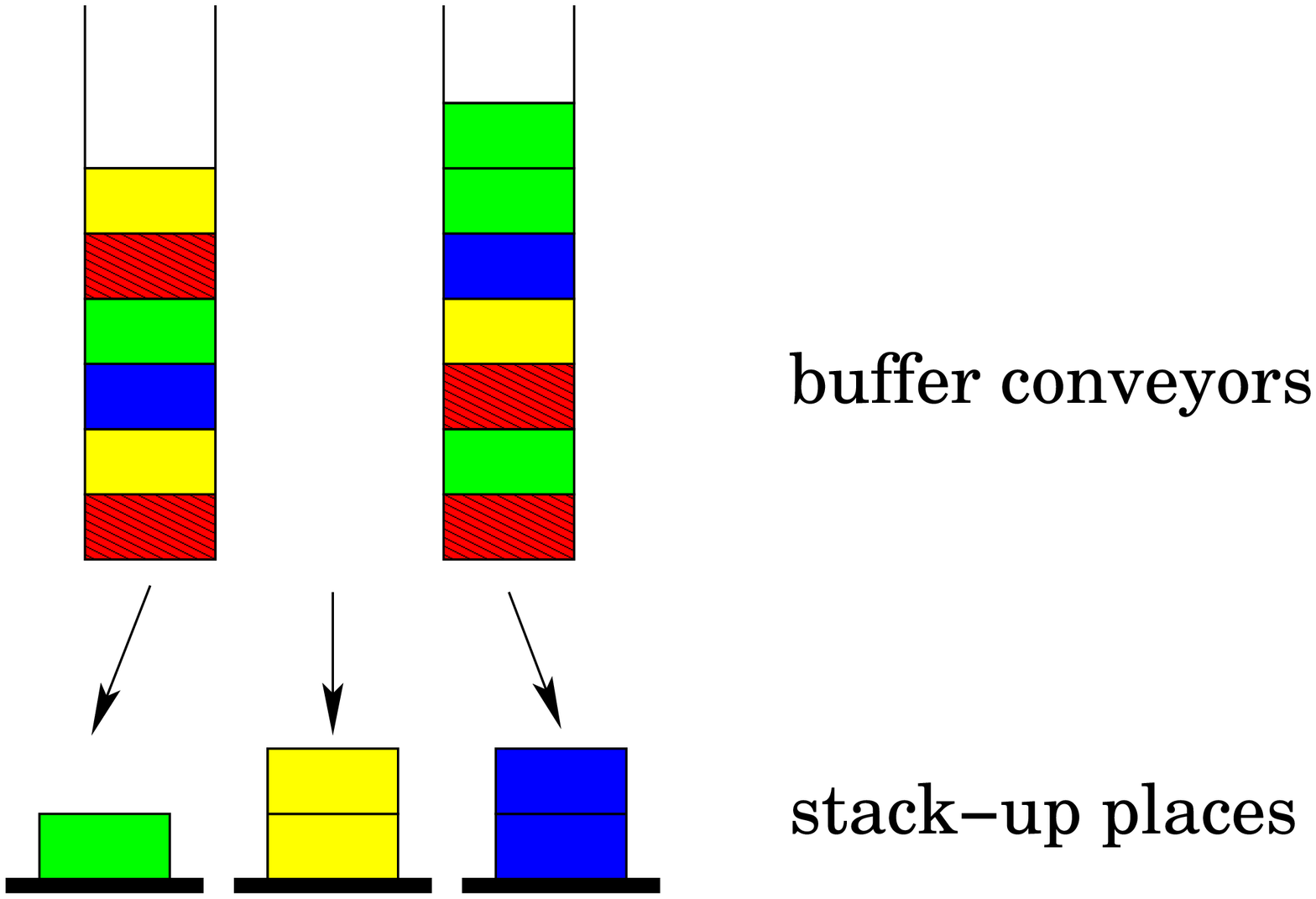}}
\caption{The FIFO stack-up system analyzed in this paper.}
\label{F01}}
\end{figure}

%\begin{figure}[h]
%\centerline{\epsfxsize=65mm \epsfbox{fifo.eps}}
%\caption{A multi-line stack-up system with a pre-placed cyclic storage conveyor.}
%\label{F00}
%\end{figure}

If we ignore the task to distribute the bins from the main
conveyor to the $k$ buffer conveyors, i.e., if the  buffer conveyors are already filled,
and if each arm can only pick-up the first bin of one of the buffer conveyors, as it is
the case if stacker cranes are used,
then the system is called a {\em FIFO palletizing system}.
Such systems can be modeled by several simple queues, see Figure \ \ref{F01}.

%\begin{figure}[h]
%\centerline{\epsfxsize=65mm \epsfbox{stack-up2.eps}}
%\caption{The FIFO stack-up system analyzed in this paper. 
%The system is blocked (cf.~page~\pageref{page-blocking} for the definition), 
%because the pallet for the red bins cannot be opened and the pallets for the 
%green, yellow, and blue bins cannot be closed.}
%\label{F01}
%\end{figure}

From a theoretical point of view, an instance of the {\sc FIFO Stack-Up} problem
consists of $k$ sequences $q_1, \ldots, q_k$ of bins and a number of available
stack-up places $p$. Each bin of $q$ is destined for exactly one pallet.
The {\sc FIFO Stack-Up}  problem is to decide whether one can remove iteratively the bins
from the sequences such that in each step only one of the first bins of $q_1, \ldots, q_k$
is removed and after each step at most $p$ pallets are open.
A pallet $t$ is called {\em open}, if at least
one bin for pallet $t$ has already been removed from the sequences,
and if at least one bin for pallet $t$ is still contained
in the remaining sequences. If a bin $b$ is removed from a sequence
then all bins located behind $b$ are moved-up one position to the front
(cf.~Section \ref{sec-prel} for the formal definition).

Every processing should absolutely avoid blocking situations.\label{page-blocking}
A system is {\em blocked}, if all stack-up places are occupied by
pallets, and none of the bins that may be used in the next step are destined for an open
pallet. To unblock the system,
bins have to be picked-up manually and moved to pallets by human workers.
Such a blocking situation is sketched in Figure \ref{F01}.

The single-line stack-up problem can be defined in the same way. An
instance for the single-line stack-up problem consists of one sequence
$q$ of bins, a storage capacity $s$, and a number of available stack-up places $p$.
In each step one of the first $s$ bins of $q$ can be removed. 
Everything else is defined analogously.

Many facts are known about single-line stack-up systems \cite{RW97,RW00,RW01}.
In \cite{RW97} it is shown that the single-line stack-up decision problem
is NP-complete, but can be solved efficiently if the storage capacity $s$
or the number of available stack-up places $p$ is a fixed constant. The problem remains NP-complete as shown in
\cite{RW00}, even if the sequence contains at most 9 bins per pallet.
In \cite{RW00}, a polynomial-time off-line approximation algorithm
for minimizing the storage capacity $s$ is introduced.
This algorithm yields a solution that is optimal up to a factor bounded by $\log(p)$.
In \cite{RW01} the performances of simple {\em on-line} stack-up
algorithms are compared with optimal off-line solutions by a competitive
analysis \cite{Bor98,FW98}. 

%In this paper, we consider the {\sc FIFO Stack-Up} problem, i.e.,
%the multi-line stack-up problem with storage capacity $s=1$.
The  {\sc FIFO Stack-Up} problem has  not been studied by
other authors up to now, although stack-up systems play an important
role in delivery industry and warehouses.
In our studies, we neither limit the number of bins for a pallet nor restrict
the number of stack-up places to the number of buffer conveyors.
That is, the number of stack-up
places can be larger than or less than the number of buffer conveyors.
We could show in \cite{GRW16b} that the {\sc FIFO Stack-Up} problem is NP-complete, 
but can be solved in polynomial time
if the number $k$ of sequences or the number $p$ of stack-up places
is fixed. 

%
%The hardness was shown in MMOR
%
%

This paper is organized as follows.
In Section \ref{sec-prel}, we give preliminaries and the problem statement.
In Section \ref{sec-hard}, we recall the definition of a {\em sequence graph} from \cite{GRW16b}.
This digraph has a vertex for every pallet and an arc from pallet $a$ to pallet $b$,
if and only if in any processing pallet $b$ can only be closed if pallet $a$
has already been opened. We give an algorithm which allows us to compute the sequence graph in 
time $\bigo(n + k\cdot m^2)$, 
where $m$ represents the number of 
pallets and $n$ denotes the total number of bins in all sequences.
%
%This graph is very useful, since
%there is a processing of some list of sequences with at most 
%$p$ stack-up places if and only if the sequence graph of this list has 
%directed pathwidth at most  $p-1$, see \cite{GRW16b}. 
%
%This connection implies that the {\sc FIFO Stack-Up} problem is NP-complete
%in general, even if there are at most 6 bins for every pallet. 
%
Further we show that the hardness of the {\sc FIFO Stack-Up} problem
even holds for the practical case $k<m$ and we
the discuss the influence of the distribution
of the pallets onto the sequences.

In Section \ref{se} we consider two further 
digraph models that allow us to find dynamic programming
solutions for the {\sc FIFO Stack-Up} problem. 
The first digraph is the {\em processing graph}.
It has a vertex for every possible configuration of the system and an arc from configuration $A$
to configuration $B$ if configuration $B$ can be obtained from configuration $A$
by a single processing step. The algorithmic use of the 
processing graph was already mentioned in \cite{GRW16b} and will be
explained more in detail here in order to ease the understanding 
on the second digraph model, which is called the {\em decision graph}.
It has a vertex for every decision configuration of the system, i.e.  
for configurations such that for every sequence the next bin  
is destined for a {\em non-open} pallet and  an arc from decision configuration $A$
to decision configuration $B$ if configuration $B$ can be obtained from configuration $A$
by a processing step including automatic steps.  The decision graph allows us to 
give a breadth first search solution for the {\sc FIFO Stack-Up} problem
of running time $\bigo(n^2 \cdot (m+2)^k)$.

In Section \ref{sec-algo} we apply methods to solve hard problems to the  {\sc FIFO Stack-Up} 
problem. Therefor we consider 
restricted versions of the problem,
exponential time algorithms, and
approximation algorithms.
We also give two integer programming models to solve the  {\sc FIFO Stack-Up} problem.
The first model computes a bin solution, i.e. an order in which 
the bins can be removed, and the second one computes a pallet solution, i.e. an 
order in which the pallets can be opened.
Both models are using a polynomial number of variables 
and a polynomial number
of constraints to compute the minimum number of stack-up places.
Further we study the
fixed-parameter tractability of the problem. The idea behind
fixed-parameter tractability is to split the complexity into
two parts - one part that depends purely on the size of the
input, and one part that depends on some {\em parameter} of the
problem that tends to be small in practice. 
Based on our three digraph models and our two integer programming models 
we give parameterized algorithms
for various parameters which imply efficient solutions for the {\sc FIFO Stack-Up} 
problem restricted to small parameter values.  
 
In Section \ref{sec-exp} we introduce a method to generate random, but realistic 
instances for the {\sc FIFO Stack-Up} problem. 
We generated instances on several thousand bins of the {\sc FIFO Stack-Up} problem
which could be solved by
our breadth first search solution combined with some cutting technique 
on the decision graph in a few minutes.
Further we analyze two integer programming approaches implemented in CPLEX 
and GLPK.
%Our two linear programming
%approaches can only be used to handle instances up to 100 bins and less than 10 pallets.
As expected CPLEX can solve the instances much faster than GLPK and 
our pallet solution approach is much better than the bin solution approach.

%%%%%%%%%%%%%%%%%%%%%%%%%%%%%%%%%%%%%%%%%%%%%%%%%%%%%%%%%%%%%%%%%%%%%%%%%%
\section{Problem Statement}\label{sec-prel}
%%%%%%%%%%%%%%%%%%%%%%%%%%%%%%%%%%%%%%%%%%%%%%%%%%%%%%%%%%%%%%%%%%%%%%%%%%

We consider {\em sequences} 
$$q_1=(b_{1}, \ldots, b_{n_1}), \ldots, q_\ell=(b_{n_{\ell-1}+1}, \ldots,b_{n_\ell}), \ldots, q_k=(b_{n_{k-1}+1}, \ldots,b_{n_k})$$ of {\em bins}. All these bins are
pairwise distinct. These sequences represent the buffer queues (handled by the
buffer conveyors) in real stack-up systems.
Each bin $b$ is labeled with a {\em pallet symbol} $\PL(b)$ which
can be some positive integer. We say bin $b$ is
destined for pallet $\PL(b)$.
The labels of the pallets can be chosen arbitrarily, because we only need
to know whether two bins are destined for the same pallet or for different
pallets. The set
of all pallets of the bins in some sequence $q_i$ is denoted by
$$\PLS(q_i)= \{ \PL(b) ~|~ b \in q_i \}.$$
For a list of sequences $Q = (q_1, \ldots, q_k)$ we denote
$$\PLS(Q) = \PLS(q_1) \cup \cdots \cup \PLS(q_k).$$
For some sequence $q = (b_1, \ldots, b_n)$ we say bin $b_i$ is {\it on the
left of} bin $b_j$ in sequence $q$ if $i < j$. And we say that such a bin
$b_i$ is on the {\it position} $i$ in sequence $q$, i.e.\ there are $i-1$
bins on the left of $b$ in sequence $q$.  The position of the first
bin in some sequence $q_i$ destined for some pallet $t$ is denoted by
$\FI(q_i,t)$, similarly the position of the last bin for pallet $t$ in
sequence $q_i$ is denoted by $\LA(q_i,t)$. For technical reasons, if 
there is no bin for pallet $t$
contained in sequence $q_i$, then we define $\FI(q_i,t) = |q_i|+1$, and
$\LA(q_i,t) = 0$.\label{label-page-fi-and-la}

Let $Q = (q_1, \ldots, q_k)$ be a list of sequences, and let
$C_Q = (i_1, \ldots, i_k)$ be some tuple in $\NN^k$. Such a tuple $C_Q$
is called a {\em configuration}, if $0 \leq i_j \leq |q_j|$ for each sequence
$q_j \in Q$.\footnote{An alternative definition of configurations
using subsequences was given in \cite{GRW16b}.} 
Value $i_j$ denotes the number of bins that have been removed from
sequence $q_j$, see Example \ref{EX1}.

A pallet $t$ is called {\em open}\label{def-open} in configuration $C_Q = (i_1, \ldots, i_k)$,
if there is a bin for pallet $t$ at some position less than or equal
to $i_j$ in sequence $q_j$, i.e. $\FI(q_j,t)\leq i_j$,
and if there is another bin for pallet
$t$ at some position greater than $i_\ell$ in sequence $q_\ell$, i.e. $\LA(q_\ell,t)>i_\ell$, 
see Example \ref{EX1}.
In view of the practical background we only consider sequences that contain
at least two bins for each pallet.
The {\em set of open pallets} in configuration $C_Q$ is denoted by $\OPEN(C_Q)$,
the number of open pallets is denoted by $\# \OPEN(C_Q)$.

\begin{remark}\label{remark-open}
Within several algorithms we will need the set of open pallets within some
configuration $C_Q$. 
Therefore we first compute  all the values $\FI(q_i,t)$ and $\LA(q_i,t)$ for $q_i\in Q$ and
$t\in \PLS(Q)$ in time $\bigo(k\cdot \max\{|q_1|, \ldots, |q_k|\}  + k\cdot m)$ respectively 
$\bigo(n + k\cdot m)$. Using these values we can test in time 
$\bigo(k)$ whether some pallet $t$ is open within a given configuration.
By performing this test for each of the $m$ pallets, for some configuration $C_Q$,
we can compute $\# \OPEN(C_Q)$ in time $\bigo(m\cdot k)$. 
\end{remark}

A pallet $t \in \PLS(Q)$ is called {\em closed} in configuration $C_Q$, if 
$\LA(q_j,t) \leq i_j$ for each sequence $q_j \in Q$.
Initially all pallets are {\em unprocessed}.
From the time when the first bin of a pallet $t$ has been removed from a sequence,
pallet $t$ is either open or closed.

For some configuration $(i_1,\ldots,i_k)$ we define
$$\FRONT((i_1,\ldots,i_k))=\{\PL(b) ~|~ 1\leq j\leq k, b \text{ is on position } i_j+1 \text{ in sequence } q_j\}.$$
Informally speaking $\FRONT((i_1,\ldots,i_k))$ is the set of all pallets of the first bins of the remaining sequences 
in configuration  $(i_1,\ldots,i_k)$. (cf.  Example \ref{EX1} and Table \ref{TB1}).

Let $C_Q = (i_1, \ldots, i_k)$ be a configuration. The removal of the bin on
position $i_j+1$ from sequence $q_j$ is called a {\em transformation step}. A
sequence of transformation steps that transforms the list $Q$ of $k$ sequences
from the initial configuration $(0,0, \ldots, 0)$ into the final configuration
$(|q_1|, |q_2|, \ldots, |q_k|)$ is called a {\em processing} of $Q$, see
Example \ref{EX1}.

It is often convenient to use pallet identifications instead of
bin identifications to represent a sequence $q$. For $r$ not
necessarily distinct pallets $t_1, \ldots, t_r$ let $[t_1,\ldots,t_r]$
denote some sequence of $r$ pairwise distinct bins $(b_1, \ldots, b_r)$,
such that $\PL(b_i) = t_i$ for $i=1,\ldots,r$.
We use this notation for lists of sequences as well. Furthermore, for some positive
integer value $n$, let $[n] := \{1,2,\ldots, n\}$  be 
the set of all positive integers between $1$ and $n$.\footnote{We will use square brackets 
in several different notations. Although the meaning becomes clear from the context 
we want to emphasize this fact.}

\begin{example}[Processing]\label{EX1}
Consider the list $Q = (q_1, q_2)$ of two sequences
$$q_1 = (b_{1},\ldots,b_{4}) = [a,b,a,b]$$ 
and
$$q_2 = (b_{5},\ldots,b_{10}) = [c,d,c,d,a,b].$$
Table \ref{TB1} shows a processing of $Q$ with 2 stack-up places.
The underlined bin is always the bin that will be removed in the
next transformation step. The already removed bins are shown greyed out.
\end{example}

\begin{table}[hbtp]
\[\begin{array}{|r|lll|c|l|l||c|}
\hline
i & q_1 & ~ & q_2 & C_Q & \FRONT(C_Q) & \OPEN(C_Q) & \mbox{bin to remove} \\
\hline
 0 & [a,b,a,b] & ~ & [\underline{c},d,c,d,a,b] & (0,0)& \{a,c\} & \emptyset  & b_{5} \\
\hline
 1 & [a,b,a,b] & ~ & [{\gr c,\,} \underline{d},c,d,a,b]  & (0,1) & \{a,d\}& \{c\}  & b_{6}\\
\hline
 2 & [a,b,a,b] & ~ & [{\gr c,d,\,} \underline{c},d,a,b]& (0,2) &\{a,c\}   & \{c,d\}  & b_{7}\\
\hline
 3 & [a,b,a,b] & ~ & [{\gr c,d,c,\,} \underline{d},a,b] & (0,3) & \{a,d\}& \{d\}  & b_{8} \\
\hline
 4 & [\underline{a},b,a,b] & ~ & [{\gr c,d,c,d,\,} a,b]& (0,4)&  \{a\}   & \emptyset & b_{1}\\
\hline
 5 & [{\gr a,\,} b,a,b] & ~ & [{\gr c,d,c,d,\,} \underline{a},b] & (1,4) & \{a,b\}  & \{a\} & b_{9}\\
\hline
 6 & [{\gr a,\,} \underline{b},a,b] & ~ & [{\gr c,d,c,d,a,\,} b] & (1,5)& \{b\}   & \{a\} & b_{2}\\
\hline
 6 & [{\gr a,b,\,} a,b] & ~ & [{\gr c,d,c,d,a,\,} \underline{b}] & (2,5)&\{a,b\}  & \{a,b\} & b_{10}\\
\hline
 7 & [{\gr a,b,\,} \underline{a},b] & ~ & [{\gr c,d,c,d,a,b} ]& (2,6) &\{a\} & \{a,b\}& b_{3}  \\
\hline
 8 & [{\gr a,b,a,\,} \underline{b}] & ~ & [{\gr c,d,c,d,a,b} ] & (3,6)&\{b\}  & \{b\} & b_{4}\\
\hline
 9 & [{\gr a,b,a,b} ] & ~ & [{\gr c,d,c,d,a,b} ]& (4,6) &\emptyset  & \emptyset & - \\
\hline
\end{array}\]
\caption{A processing of $Q = (q_1, q_2)$ from Example \ref{EX1}
with 2 stack-up places. In this simple example it is easy to see that
there is no processing of $Q$ that needs less than 2 stack-up places, because
pallets $a$ and $b$ as well as $c$ and $d$ are interlaced.}
\label{TB1} 
\end{table}

We consider the following problem.

\begin{desctight} 
\item[Name] {\sc FIFO Stack-Up}

\item[Instance] 
A list $Q = (q_1, \ldots, q_k)$ of $k$ sequences of bins, 
for every bin $b$ of $Q$ its pallet symbol $\PL(b)$, and a positive
integer $p$.

\item[Question] 
Is there a processing of $Q$, such that in each configuration during
the processing of $Q$ at most $p$ pallets are open?
\end{desctight}

We use the following variables in the analysis of our algorithms:
$k$ denotes the number of sequences, and
$p$ stands for the number of stack-up places, while
$m$ represents the number of pallets in $\PLS(Q)$, and
$n$ denotes the total number of bins, i.e. $n=n_k$. Finally,
$N =\max\{|q_1|, \ldots, |q_k|\}$ is the maximum sequence length.

For some instance $I$ of the {\sc FIFO Stack-Up} problem numbers are encoded
binary and sequences are encoded by sequences of pallet symbols thus
the size $|I|$ can be bounded by
\begin{equation}
|I| \in \bigo(n \cdot \log_2(m) + \log_2(p)).\label{eq-inst}
\end{equation}

The size of the input is important for the analysis of running times 
in Section \ref{sec-algo}.

\begin{remark}\label{remark-practical-br}
In view of the practical background, it holds $p < m$, otherwise
each pallet could be stacked-up onto a different stack-up place. Furthermore,
$k < m$, otherwise all bins of one pallet could be channeled into
one buffer queue in the multi-line stack-up systems with pre-placed cyclic storage
conveyor, see Figure \ref{F00}.
Finally $m \leq \frac{n}{2} < n$, since there are at least two bins for each pallet.
\end{remark}

The relation $n\leq k \cdot N$ and the assumption $m\leq \frac{n}{2}$ imply the following bound.

%\begin{cor}\label{cor-bd}
%\begin{inparaenum}[(1)]
%\item $m\leq \frac{k\cdot N}{2}$, i.e. $m\in \bigo(k\cdot N)$,
%\item $k\cdot N \leq (N+1)^k$
%\end{inparaenum}
%\end{cor}

\begin{corollary}\label{cor-bd1}
$m\leq \frac{k\cdot N}{2}$, i.e. $m\in \bigo(k\cdot N)$,
\end{corollary}

The following estimation can be shown by induction on $k$.

\begin{corollary}\label{cor-bd2}
$k\cdot N \leq (N+1)^k$
\end{corollary}

%
%Bin Sol + Pallet solution def + am bsp
%

For the solution of the {\sc FIFO Stack-Up} problem for some list of sequences $Q$ we  use the 
notation of a bin solution and of a pallet solution, which are defined as follows.
%Consider a processing of a list $Q$ of sequences. 
Let $B = (b_{\pi(1)},
\ldots, b_{\pi(n)})$ be the order in which the bins are removed during
a processing of $Q$. Then $B$ is called a {\em bin solution} of $Q$.
In Example \ref{EX1}, we have
\[
B = (b_{5}, b_{6}, b_{7}, b_{8}, b_{1}, b_{9}, b_{2},
        b_{10}, b_{3}, b_{4})
  = [c,d,c,d,a,a,b,b,a,b]
\]
as a bin solution.

Let $T = (t_1, \ldots, t_m)$ be the order in which the pallets are opened
during the processing of $Q$. Then $T$ is called a {\em pallet solution}
of $Q$. In Example \ref{EX1} we have $$T = (c,d,a,b)$$ as a pallet solution.

%%%%%%%%%%%%%%%%%%%%%%%%%%%%%%%%%%%%%%%%%%%%%%%%%%%%%%%%%%%%%%%%%%%%%%%%%%
\section{NP-hardness}\label{sec-hard}
%%%%%%%%%%%%%%%%%%%%%%%%%%%%%%%%%%%%%%%%%%%%%%%%%%%%%%%%%%%%%%%%%%%%%%%%%%

Next we recall the connection
between the used number of stack-up places for a processing of an instance $Q$ 
and the directed pathwidth of the sequence graph $G_Q$
from \cite{GRW16b}, which is useful for
our hardness results, our integer programming model for computing a pallet 
solution, and for several parameterized algorithms.

%%%%%%%%%%%%%%%%%%%%%%%%%%%%%%%%%%%%%%%%%%%%%%%%%%%%%%%%%%%%%%%%%%%%%%%%%%
\subsection{Directed Pathwidth} \label{SCdpw}
%%%%%%%%%%%%%%%%%%%%%%%%%%%%%%%%%%%%%%%%%%%%%%%%%%%%%%%%%%%%%%%%%%%%%%%%%%

%Next we consider a useful relation between 
%an instance of the {\sc FIFO Stack-Up} problem and the directed
%pathwidth of a directed graph model. 
According to Bar{\'a}t \cite{Bar06}, the notion of directed pathwidth was
introduced by Reed, Seymour, and Thomas around 1995 and relates to directed
treewidth introduced by Johnson, Robertson, Seymour, and Thomas in
\cite{JRST01}. A directed path-decomposition of a digraph $G=(V,A)$
is a sequence $(X_1, \ldots, X_r)$ of subsets of $V$, called {\em bags}, such 
that the following three conditions hold true.
\begin{enumerate}[(1)]
\item $X_1 \cup \ldots \cup X_r ~=~ V$, 
\item for each $(u,v) \in A$ there is a pair $i \leq j$ such that
  $u \in X_i$ and $v \in X_j$, and 
\item for all $i,j,\ell$ with $1 \leq i < j < \ell \leq r$ it holds
  $X_i \cap X_\ell \subseteq X_j$. 
\end{enumerate}
The {\em width} of a directed path-decomposition ${\cal X}=(X_1, \ldots, X_r)$ 
is $\max_{1 \leq i \leq r} |X_i|-1$. The {\em directed pathwidth} of $G$,
$\dpw(G)$ for short, is 
the smallest integer $w$ such that there is a directed path-de\-com\-po\-sition for 
$G$ of width $w$. For symmetric digraphs, the directed pathwidth is equivalent 
to the undirected pathwidth of the corresponding undirected graph \cite{KKKTT12}, which
implies that determining
whether the pathwidth of some given digraph  is 
at most some given value $w$ is NP-complete. 
For each constant $w$, it is decidable in polynomial time whether a given 
digraph has directed pathwidth at most $w$, see Tamaki \cite{Tam11}.

%%%%%%%%%%%%%%%%%%%%%%%%%%%%%%%%%%%%%%%%%%%%%%%%%%%%%%%%%%%%%%%%%%%%%%%%%%
\subsection{The Sequence Graph} \label{SCgb}
%%%%%%%%%%%%%%%%%%%%%%%%%%%%%%%%%%%%%%%%%%%%%%%%%%%%%%%%%%%%%%%%%%%%%%%%%%

The sequence graph $G_Q = (V,A)$ for an instance $Q=(q_1, \ldots ,q_k)$
is defined by vertex set $V = plts(Q)$ and the following set of arcs.
There is an arc $(u,v) \in A$ if and only if there is a sequence
$q_\ell = (b_{n_{\ell-1}+1}, \ldots, b_{n_\ell})$ with two bins $b_i$, $b_j$ such that
$i < j$,
$\PL(b_i)=u$,
$\PL(b_j)=v$, and
$u \neq v$.

%We next recall an example from \cite{GRW16b} and proceed 
%with a method which shows how to create the sequence graph.

\begin{example}[Sequence Graph] \label{EX6}
Figure \ref{F04} shows the sequence graph $G_Q$ for $Q = (q_1, q_2, q_3)$ 
with sequences $q_1 = [a,a,d,e,d]$, $q_2 = [c,b,b,d]$, and 
$q_3 = [c,c,d,e,d]$.
\end{example}

\begin{figure}[ht]
\centerline{\epsfxsize=40mm \epsfbox{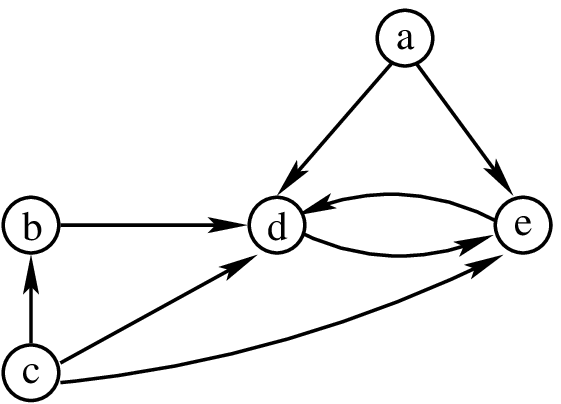}}
\caption{Sequence graph $G_Q$ of Example \ref{EX6}.}
\label{F04}
\end{figure}

If $G_Q = (V,A)$ has an arc $(u,v) \in A$ then $u \neq v$ and for
every processing of $Q$, pallet $u$ is opened before pallet $v$ is
closed. Digraph $G_Q = (V,A)$ can be computed in 
time 
%$\bigo(n + k\cdot |A|) \subseteq 
$\bigo(n + k\cdot m^2)$ by the algorithm
{\sc Create Sequence Graph} shown in Figure \ref{fig:algorithm5}.

A value is added to a list only if it is not already contained. To check
this efficiently in time $\bigo(1)$ we have to use a boolean array. 
In our algorithm $V$ and $L$ are implemented by boolean arrays.
Therefore we need some preprocessing phase where we run through each sequence and seek
for the pallets. This can be done in time $\bigo(n+k\cdot m)\subseteq \bigo(n+m^2)$.

\begin{figure}[ht]
\hrule
{\strut\footnotesize \bf Algorithm {\sc Create Sequence Graph}} 
\hrule
\begin{tabbing}
xxxx \= xxxx \= xxxx \= xxxx \= xxxx \= xxxx \= xxxx \=\kill
for each sequence $q \in Q$ do \\
\> $b :=$ first bin of sequence $q$ \\
\> add $\PL(b)$ to vertex set $V$, if it is not already contained \\
\> $L := (\PL(b))$ \>\>\>\>\>\>   $\vartriangleright$ $L$ contains pallets of bins up to bin $b$ \\
\> for $i := 2$ to $|q|$ do \\
\>\> $b := i$-th bin of sequence $q$ \\
\>\> add $\PL(b)$ to vertex set $V$, if it is not already contained \\
\>\> if ($i=\LA(q,\PL(b))$)\\
\>\>\> for each pallet $t \in L$ do \\
\>\>\>\>  if $t \neq \PL(b)$ add arc $(t,\PL(b))$ to arc set $A$, if it is not already contained \\
\>\> if ($i=\FI(q,\PL(b))$)\\
\>\>\>  $append(\PL(b),L)$ 
\end{tabbing}
\hrule
\caption{Create the sequence graph $G=(V,A)$ for some given list of sequences $Q$.}
\label{fig:algorithm5}
\end{figure}

In \cite{GRW16b} we have shown the following correlation 
between the used number of stack-up places for a processing of an instance $Q$ 
and the directed pathwidth of the sequence graph $G_Q$.

\begin{theorem}\label{th-pw-gbefore}
Let $Q = (q_1, \ldots, q_k)$ then digraph $G_Q=(V,A)$ has 
directed pathwidth at most $p-1$ if and only if $Q$ can be processed with at 
most $p$ stack-up places.
\end{theorem}

This characterization extends the previously known 
areas of applications for directed pathwidth  
in graph databases and boolean networks, which have been shown in \cite{DE14}.

%%%%%%%%%%%%%%%%%%%%%%%%%%%%%%%%%%%%%%%%%%%%%%%%%%%%%%%%%%%%%%%%%%%%%%%%%%%%%%%%
\subsection{Hardness Result}\label{sec-hardness}
%%%%%%%%%%%%%%%%%%%%%%%%%%%%%%%%%%%%%%%%%%%%%%%%%%%%%%%%%%%%%%%%%%%%%%%%%%%%%%%%

Next we will discuss the hardness of the {\sc FIFO Stack-Up}  problem. In contrast to 
Section \ref{SCgb} we transform an instance of a graph problem into 
an instance of the {\sc FIFO Stack-Up}  problem.

%Let $G=(V,A)$ be a digraph. We will assume that $G=(V,A)$
%does not contain any vertex with only  outgoing arcs
%and not contain any vertex with only incoming arcs.
%This is only for technical 
%reasons and the removal of such vertices 
%will not change the directed pathwidth of $G$, because 
%a vertex $u$ with only outgoing arcs can be 
%placed in a singleton $X_i=\{u\}$ at the beginning
%of the directed path-decomposition and a vertex $u$ with only incoming arcs can be 
%placed in a singleton $X_i=\{u\}$ at the end
%of the directed path-decomposition, without to change its width.

Let $G=(V,A)$ be some  digraph and $A=\{a_1,\ldots,a_\ell\}$ its arc set.
The {\em sequence system} $Q_G = (q_1,\ldots,q_\ell)$ for $G$ is defined as follows.
\begin{enumerate}[(1)]
\item There are $2\ell$ bins $b_1,\ldots,b_{2\ell}$.

\item Sequence $q_i=(b_{2i-1},b_{2i})$ for $1\leq i \leq \ell$.

\item The pallet symbol of bin $b_{2i-1}$ is the first vertex of arc $a_i$
and the pallet symbol of $b_{2i}$ is the second vertex of arc $a_i$ for
$1\leq i \leq \ell$. Thus  $\PLS(Q_G) = V$.
\end{enumerate}

\begin{example}[Sequence System]\label{EX3-new}
For the digraph $G$ of Figure \ref{F-ex-d}
the corresponding sequence system is $Q_G=(q_1,q_2,q_3,q_4,q_5,q_6,q_7)$,
where
$$
\begin{array}{lcllcllcllcllcl}
q_1 &= &[a,b], & q_2 &=& [b,c], & q_3 &=& [c,d], &q_4 &=& [d,e], &q_5 &= &[e,a],    \\
q_6 &=& [e,f], & q_7 &=& [f,a].
\end{array}$$
The sequence graph of $Q_G$ is $G$.
\end{example}

\begin{figure}[hbtp]
\centerline{\epsfxsize=45mm \epsfbox{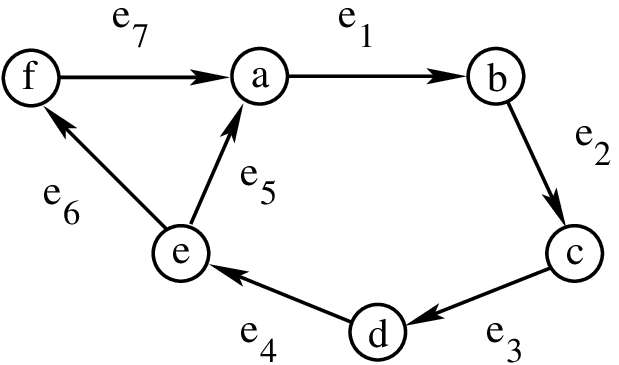}}
\caption{Digraph $G$ of Example \ref{EX3-new}.}
\label{F-ex-d}
\end{figure}

The definition of sequence system $Q_G$ and sequence graph 
$G_Q$, defined in Section \ref{SCgb},  imply the following 
results shown in \cite{GRW16b}.

\begin{proposition}[\cite{GRW16b}]\label{prop}
For every digraph $G$ it holds $G = G_{Q_G}$.
\end{proposition}

\begin{theorem}[\cite{GRW16b}]\label{cor-hard}
The {\sc FIFO Stack-Up} problem 
is NP-complete, even if the
sequences of $Q$ contain together at most 6 bins per pallet.
\end{theorem}

Next we strengthen the hardness by considering further
restricted versions of the {\sc FIFO Stack-Up} problem.
First we can bound the maximum
sequence length $N$, since by the definition of the sequence system we obtain 
instances where $N=2$ in the proof of Theorem \ref{cor-hard}.

\begin{corollary}\label{cor-hard2}
The {\sc FIFO Stack-Up} problem 
is NP-complete, even if $N$ is bounded by some constant greater than 1.
For $N=1$ the {\sc FIFO Stack-Up} problem  can obviously be solved in polynomial time.
\end{corollary}

In order to consider the distribution of the bins of a pallet $t$
onto the sequences we define
$$d_Q(t)=|\{q\in Q ~|~ t\in \PLS(q)\}|$$
and
$$d_Q=\max_{t\in[m]}d_Q(t).$$

By Theorem \ref{cor-hard} we have shown the following result.

\begin{corollary}\label{cor-hard3}
The {\sc FIFO Stack-Up} problem 
is NP-complete, even if $d_Q=6$.
For $d_Q=1$ the  {\sc FIFO Stack-Up} problem  can be solved in polynomial time,
since we can process all sequences one after the other. 
\end{corollary}

For $d_Q\in\{2,\ldots,5\}$
the complexity of the  {\sc FIFO Stack-Up} problem
remains open.

In Remark \ref{remark-practical-br}  we have restricted to practical instances where
$k < m$. This is not given within the hardness results of \cite{GRW16b}. 
%In the general case of Theorem \ref{\cite{GRW16b}} we have $k\in \bigo(m^2)$ and in the 
%restricted case of Theorem \ref{cor-hard} we have $k\in \bigo(m)$.

\begin{corollary}\label{remark-practical-br2}
The {\sc FIFO Stack-Up} problem 
is NP-complete, even if $k < m$ and the
sequences of $Q$ contain together at most 6 bins per pallet.
\end{corollary}

\begin{proof}
In order to carry over the hardness to $k < m$, we can modify a given list of
sequences $Q=(q_1,\ldots,q_k)$ as follows. We introduce $3k$ new pallets
$a_i,b_i,c_i$ for $1\leq i \leq k$. We define new sequences $q_i$, $1\leq i \leq k$, 
from the old ones  by 
$$q'_i=q_i \circ [a_i,a_i,b_{i},b_{i},c_{i},c_{i}],$$
i.e. by concatenating $q_i$ and $[a_i,a_i,b_{i},b_{i},c_{i},c_{i}]$.
Since list $Q$ can be processed
with at most $p$ stack-up places, if and only if list $Q'=(q'_1,\ldots,q'_k)$ can be processed
with at most $p$ stack-up places and the number of pallets in $Q'$ can be increased
arbitrarily, the {\sc FIFO Stack-Up} problem is also hard for the case $k < m$. \qed
\end{proof}

%%%%%%%%%%%%%%%%%%%%%%%%%%%%%%%%%%%%%%%%%%%%%%%%%%%%%%%%%%%%%%%%%%%%%%%%%%
\section{Dynamic Programming Algorithms to Solve  the  {\sc FIFO Stack-Up} Problem}\label{se}
%%%%%%%%%%%%%%%%%%%%%%%%%%%%%%%%%%%%%%%%%%%%%%%%%%%%%%%%%%%%%%%%%%%%%%%%%%

Our aim in controlling FIFO stack-up systems is to compute a processing
of the given sequences of bins with a minimum number of stack-up places.
Such an optimal processing can always be found by computing the {\it
processing graph} or the {\it decision graph}. The algorithmic use of the 
processing graph was already mentioned in \cite{GRW16b} and will next be
explained in more detail in order to ease the understanding in the subsequent  
section on the decision graph.

%%%%%%%%%%%%%%%%%%%%%%%%%%%%%%%%%%%%%%%%%%%%%%%%%%%%%%%%%%%%%%%%%%%%%%%%%%
\subsection{The Processing Graph}\label{SC2}
%%%%%%%%%%%%%%%%%%%%%%%%%%%%%%%%%%%%%%%%%%%%%%%%%%%%%%%%%%%%%%%%%%%%%%%%%%

The processing graph $G=(V,A)$ contains a vertex for every possible confi\-gu\-ration.
Each vertex $v$ representing some configuration $C_Q(v)$ is labeled by
the number $\#\OPEN(C_Q(v))$.
There is an arc from vertex $u$ representing configuration $(u_1, \ldots, u_k)$
to vertex $v$ representing configuration $(v_1, \ldots, v_k)$
if and only if $u_i = v_i - 1$ for exactly one element of the vector
and for all other elements of the vector $u_j = v_j$.
The arc is labeled with the bin that will be removed
in the corresponding transformation step.

\begin{example}[Processing Graph]\label{EX1-pg}
For the sequences of Example \ref{EX1}
we get the processing graph of Figure \ref{F07}. 
\end{example}

\begin{figure}[htbp]
\centerline{\epsfxsize=.75\textwidth \epsfbox{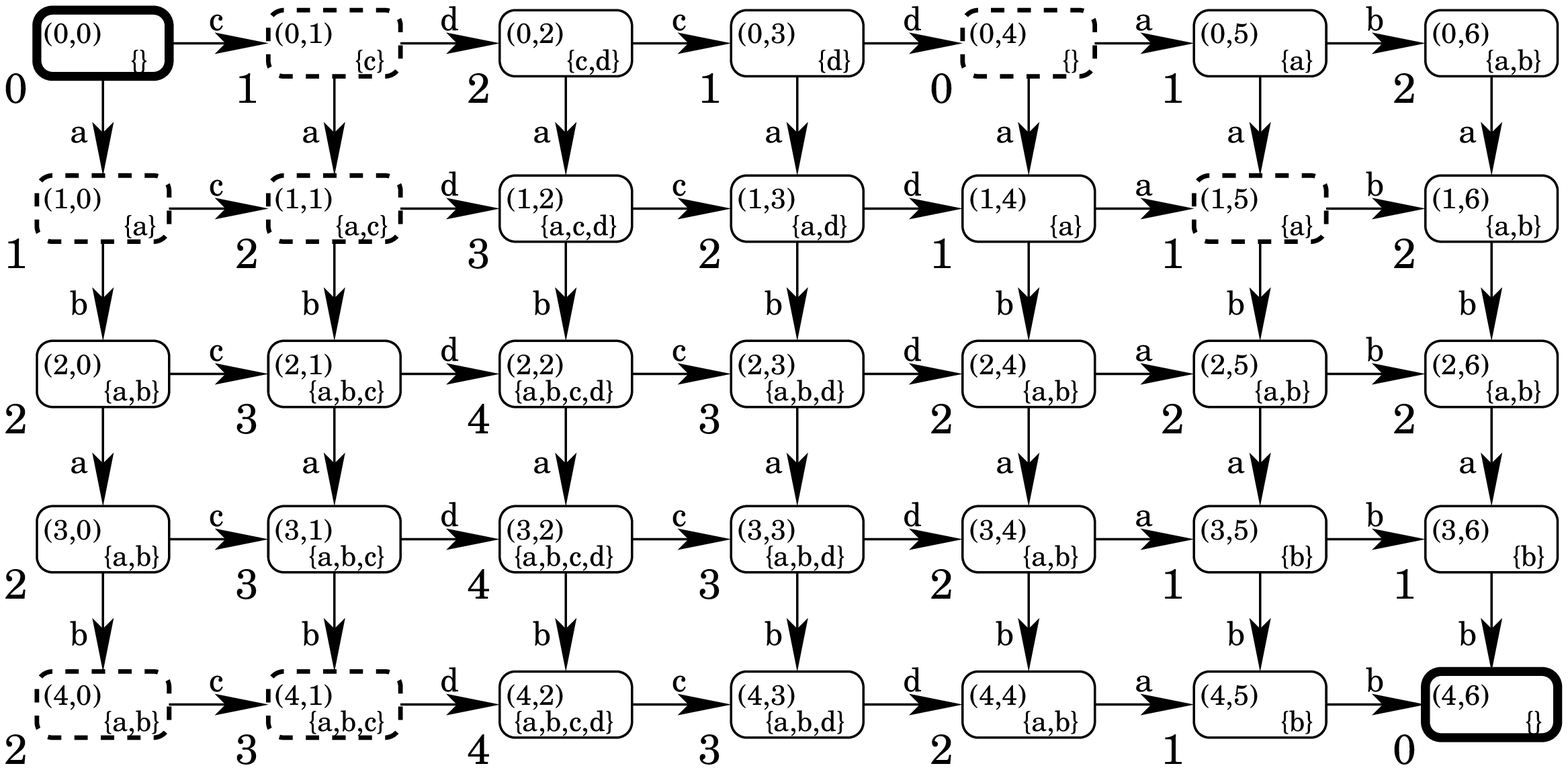}}
\caption{The processing graph of Example \ref{EX1}. Instead of the bin each arc
is labeled with the pallet symbol of the bin that will be removed in the
corresponding transformation step. The shaded vertices will be important 
in the following section. The upper left vertex represents the initial configuration
and the lower right vertex represents the final configuration.}
\label{F07}
\end{figure}

Obviously, every processing 
graph is directed and acyclic.
Every  bin solution describes a path from
the initial configuration $(0,0, \ldots, 0)$ to the final configuration
$(|q_1|, |q_2|, \ldots, |q_k|)$ in the processing graph.
We are interested in such paths where the maximal vertex label on
that path is minimal. \label{label-path-min}

The processing graph can be computed in time $\bigo(k \cdot (N+1)^k)$ by
some breadth first search algorithm as follows. We store the already discovered
configurations, i.e.\ the vertices of the graph, in some list $L$.
Initially, list $L$ contains only the initial configuration.
In each step of the algorithm we take the first configuration out
of list $L$. Let $C_Q=(i_1, \ldots, i_k)$ be such a configuration.
For each $j \in [k]$ we remove the bin on position $i_j+1$
from sequence $q_j$, and get another configuration $C'_Q$.
We append $C'_Q$ to list $L$, add it to vertex set $V$, if it is not already contained, 
and add $(C_Q, C'_Q)$ to arc set $A$.

For each configuration we want to store the number of open pallets of that
configuration. This can be done efficiently in the following way.
First, since none of the bins has been removed from any sequence in the
initial configuration, we have $\# \OPEN((0, \ldots, 0)) = 0$.  In each
transformation step we remove exactly one bin for some pallet $t$ from
some sequence $q_j$, thus
\begin{eqnarray}
  \# \OPEN((i_1, \ldots, i_{j-1}, i_j+1, i_{j+1}, \ldots, i_k)) & = &
   \# \OPEN((i_1, \ldots, i_{j-1}, i_j, i_{j+1}, \ldots, i_k)) ~+~ c_{j}
  \label{transStep}
\end{eqnarray}
where $c_{j} = 1$ if pallet $t$ has been opened in the
transformation step, and $c_{j} = -1$ if pallet $t$ has been closed
in the transformation step. Otherwise, $c_{j}$ is zero. If we
put this into a formula we get
\[
   c_{j} = \left\{
     \begin{array}{rll}
       1, & \multicolumn{2}{l}{\mbox{if } \FI(q_j,t) = i_j+1 \mbox{ and } \FI(q_\ell,t) > i_\ell
              ~~ \forall ~ \ell \neq j} \\
      -1, & \multicolumn{2}{l}{\mbox{if } \LA(q_j,t) = i_j+1 \mbox{ and } \LA(q_\ell,t) \leq i_\ell
              ~~ \forall ~ \ell \neq j} \\
       0, & \multicolumn{2}{l}{\mbox{otherwise.}}
     \end{array}
   \right.
\]
Please remember our technical definition of $\FI(q,t)$ and $\LA(q,t)$
from page \pageref{label-page-fi-and-la}
for the case that $t\not\in \PLS(q)$.

That means, the calculation of value $\# \OPEN(C_Q(v))$ for the
vertex $v$ representing configuration $C_Q(v)$ can be done in
time $\bigo(k)$ if the values $\FI(q_j,t)$ and $\LA(q_j,t)$ have
already been calculated in some preprocessing phase. Such a
preprocessing can be done in time $\bigo(k \cdot N+k\cdot m)$ due to
Remark \ref{remark-open}, which can
be bounded by $\bigo(k\cdot (N+1)^k)$ due to Corollary \ref{cor-bd1} and 
Corollary \ref{cor-bd2}. The number of vertices
is in $\bigo((N+1)^k)$, so the vertex labels can be computed in time
$\bigo(k \cdot (N+1)^k)$. Since at most $k$ arcs leave each vertex, the number
of arcs is in $\bigo(k \cdot (N+1)^k)$, and each arc can be computed
in time $\bigo(1)$.  Thus, the computing can be done in total time
$\bigo(k \cdot (N+1)^k)$.

\medskip
Let $s$ be the vertex representing the initial configuration, and
let $topol: V \to \NN$ be a topological ordering of the vertices of
the processing graph
$G=(V,A)$ such that $topol(u) < topol(v)$ holds for each $(u,v) \in A$.
For some vertex $v \in V$ and some path $P=(v_1, \ldots, v_\ell)$ with
$v_1 = s$, $v_\ell = v$ and $(v_i, v_{i+1}) \in A$ we define
\[
   val_P(v) := \max_{u \in P} (\# \OPEN(C_Q(u)))
\]
to be the maximum vertex label on that path. Let ${\mathcal P}(v)$ denote
the set of all paths from vertex $s$ to vertex $v$. Then we define
\[
   val(v) := \min_{P \in {\mathcal P}(v)} (val_P(v)).
\]
The problem is to compute the value $val(t)$ where $t$ is the vertex
representing the final configuration. It holds
\begin{eqnarray}
    val(v) = \max\{ \# \OPEN(C_Q(v)), \min_{(u,v) \in A}(val(u)) \},\label{eq2}
\end{eqnarray}
because each path $P \in {\mathcal P}(v)$ must go through some vertex $u$
with $(u,v) \in A$. So we may iterate over the preceding vertices of $v$
instead of iterating over all paths. If $\# \OPEN(C_Q(u)) < \# \OPEN(C_Q(v))$
for all preceding configurations then a pallet must have been opened in the last step to
reach configuration $C_Q(v)$.

\begin{figure}[ht]
\hrule
{\strut\footnotesize \bf Algorithm {\sc Find Path}} 
\hrule
\begin{tabbing}
xx \= xx \= xx \= xx\= xx\= xx \= xx \= xx\= xx \= xx \= xx \= xx \=\kill
$val[s] := 0$  \>\>\>\> \>\> \> \>\> \>\> \>$\vartriangleright$ Computation according to Equation (\ref{eq2})\\
for each vertex $v \neq s$ in order of $topol$ do \\
\> $val[v] := \infty$ \\
\> for every $(u,v) \in A$ do  \>\>\> \>\> \>\>\> \>\> \>$\vartriangleright$ Compute $\min_{(u,v) \in A}(val(u))=:val(v)$  \\
\>\> if ($val[u] < val[v]$) \\
\>\>\> $val[v] := val[u]$ \\
\>\>\> $path[v] := u$ \\
\> if ($val[v] < \# \OPEN[v]$) \>\>\> \>\> \> \>\> \>\> \>$\vartriangleright$ Compute  $\max\{ \# \OPEN(C_Q(v)), val(v)\}=:val(v)$\\
\>\> $val[v] := \# \OPEN[v]$
\end{tabbing}
\hrule
\caption{Finding an optimal processing by dynamic programming.}
\label{fig:algorithm2}
\end{figure}

The value $val(v_\ell)$ can be computed  by Algorithm {\sc Find Path}
given in Figure \ref{fig:algorithm2}. 
The corresponding path 
\begin{eqnarray}
P=(v_1, \ldots, v_\ell)\label{interest-path}
\end{eqnarray}
 is obtained by  $path[v_\ell]=v_{\ell-1}$, 
$path[v_{\ell-1}]=v_{\ell-2}$, $\ldots$,  $path[v_{2}]=v_{1}$. \label{label-path-min2}
For the running time
we observe that
a topological ordering of the vertices
of digraph $G$ can be found by a depth first search algorithm in time
$\bigo(|V| + |A|)$. The remaining work of algorithm {\sc Find Path} can also
be done in time $\bigo(|V| + |A|)$. In the processing graph we have $|V| \in
\bigo((N+1)^k)$, and $|A| \in \bigo(k \cdot (N+1)^k)$.

\medskip
It is not necessary to explicitly
build the processing graph to compute an optimal processing, we have done it 
just for the sake of clarity and to enhance understanding.
We combine the construction of the processing graph with the topological
sorting and the path finding by some breadth first search algorithm
{\sc Optimal Bin Solution} shown in Figure \ref{fig:algorithm3}.
Algorithm {\sc Optimal Bin Solution} uses
the following two operations.
\begin{itemize}
\item $head(L)$ yields the first entry of list $L$ and removes it from $L$.

\item $append(e,L)$ adds  element $e$ to the list $L$, if $e$ is not already contained in $L$.
\end{itemize}

\begin{figure}[ht]
\hrule
{\strut\footnotesize \bf Algorithm {\sc Optimal Bin Solution}} 
\hrule
\begin{tabbing}
xxx\=xxx\=xxx\=xxx \= xxx \= xxx \= xxx \= xxx\= xxx \= xxx \= xxx \= xxx\= xxx\=\kill
%xx\=xx\=xx\=xx \= xx \= xx \= xx \= xx\= xx \= xx\= xx  \= xx\= xx\=\kill
%$V := \emptyset$ \\
%$E := \emptyset$ \\
\#open$[~(0,\ldots,0)~]:=0$ \\
$val[~(0,\ldots,0)~] := 0$ \\
$L := (~(0,\ldots,0)~)$ \>\>\>\>\>\>\>\>\>\>\>\>\>$\vartriangleright$ List of uninvestigated configurations\\
$pred[~(0,\ldots,0)~] := \emptyset$ \>\>\>\>\>\>\>\>\>\>\>\>\>$\vartriangleright$ List of predecessors of some configuration \\
while $L$ is not empty do \\
\> $C := head(L)$\>\>\>\>\>\>\>\>\>\>\>\>$\vartriangleright$ let $C=(i_1,\ldots,i_k)$ \\
\> if ($pred[C]\neq \emptyset$)    \>\>\>\>\>\>\>\>\>\>\>\>$\vartriangleright$ all predecessors of $C$  are computed \\
\>\> {\sc Extend Path}($C$) \\
\> for $j := 1$ to $k$ do \\
\>\> $C_s := (i_1, \ldots, i_j+1, \ldots, i_k)$ \\
\>\> if ($C_s$ is not in $L$) \\
\>\>\> compute $\# \OPEN[C_s]$ according to Equation (\ref{transStep}) \\
\>\>\> $append(C_s,L)$ \\
\>\> $append(C,pred[C_s])$
%\>\> insert $C_s$ to vertex set $V$ \\
%\>\> insert $(C,C_s)$ to edge set $E$
\end{tabbing}
\hrule
\caption{Construction of the processing graph and computation of an optimal bin
 solution at once by breadth first search.}
\label{fig:algorithm3}
\end{figure}

\begin{figure}[ht]
\hrule
{\strut\footnotesize \bf Algorithm {\sc Extend Path}($C$)} 
\hrule
\begin{tabbing}
xxx\=xxx\=xxx\=xxx \= xxx \= xxx \= xxx \= xxx\= xxx \= xxx \= xxx \= xxx\= xxx\=\kill
%xx\=xx\=xx\=xx \= xx \= xx \= xx \= xx\= xx \= xx\= xx  \= xx\= xx\=\kill
%$V := \emptyset$ \\
$val[C] := \infty$   \\
for each $C_p$ in list $pred[C]$ do  \>\>\>\>\>\>\>\>\>\>\> \>\>$\vartriangleright$ Compute $val[C]$ due to Equation (\ref{eq2})\\
\> if ($val[C]  > val[C_p]$) \\
\>\> $val[C]  := val[C_p]$ \\
\>\> $path[C] := C_p$ \\
if ($val[C]  < \# \OPEN[C]$) \\
\>$val[C]  := \# \OPEN[C]$ 
\end{tabbing}
\hrule
\caption{Submethod to extend a path by one vertex $C$.}
\label{fig:algorithm3a}
\end{figure}

At the end of the processing of algorithm {\sc Optimal Bin Solution} the variable
$path$ contains a path $P$ as shown in (\ref{interest-path}) where the maximal vertex label is minimal.

The computation of all at most $(N+1)^k$ values $\# \OPEN(C_Q(v))$ 
can be performed in time $\bigo(k \cdot (N+1)^k)$. A value is
added to some list only if it is not already contained. To check
this efficiently in time $\bigo(1)$ we use a boolean array over all possible
configurations. 
This array can be initialized in time $\bigo((N+1)^k)$.
Thus, we can compute the minimal number of stack-up places
necessary to process the given {\sc FIFO Stack-Up}
problem as well as such a bin solution in time $\bigo(k \cdot (N+1)^k)$. 

\begin{theorem}\label{th-processing-gr}
The {\sc FIFO Stack-Up} problem 
can be solved in time $\bigo(k \cdot (N+1)^k)$.
\end{theorem}

%%%%%%%%%%%%%%%%%%%%%%%%%%%%%%%%%%%%%%%%%%%%%%%%%%%%%%%%%%%%%%%%%%%%%%%%%%
\subsection{The Decision Graph} \label{SCdec}
%%%%%%%%%%%%%%%%%%%%%%%%%%%%%%%%%%%%%%%%%%%%%%%%%%%%%%%%%%%%%%%%%%%%%%%%%%

During a processing of a list $Q$ of sequences there are often
configurations for which it is easy to find a bin $b$ that
can be removed such that a further processing with $p$
stack-up places is still possible. This is the case, if bin $b$ is
destined for an already open pallet.
A configuration $(i_1, \ldots, i_k)$ is called a {\em decision
configuration}, if the bin on position $i_j+1$ of sequence $q_j$ for
each $j \in [k]$ is destined for a non-open pallet, i.e. 
$front((i_1, \ldots, i_k))\cap \OPEN((i_1, \ldots, i_k))=\emptyset$. 
We can restrict
FIFO stack-up algorithms to deal with such decision configurations,
in all other configurations the algorithms automatically remove a bin
for some already open pallet.

A solution to the {\sc FIFO Stack-Up} problem can always be found by computing
the decision graph for the given instance of the problem. The decision graph
$G=(V,A)$ has a vertex for each decision configuration into which the initial
configuration can be transformed. There is an arc $(u,v) \in A$ from
a vertex $u$ representing decision configuration $(u_1, \ldots, u_k)$
to a vertex $v$ representing decision configuration $(v_1, \ldots, v_k)$
if and only if there is a bin $b$ on position $u_j + 1$ in sequence
$q_j$ such that the removal of $b$ in the next transformation step
and the execution of only automatic transformation steps afterwards
lead to decision configuration $(v_1, \ldots, v_k)$. 
Arc $(u,v)$
is labeled with the pallet symbol of the bin that will be removed in the
corresponding transformation step.

\begin{example}[Decision Graph]\label{EX1-dg}
In Figure \ref{F08}
the decision graph for Example \ref{EX1} is shown.
\end{example}

\begin{figure}[htbp]
\centerline{\epsfxsize=.35\textwidth \epsfbox{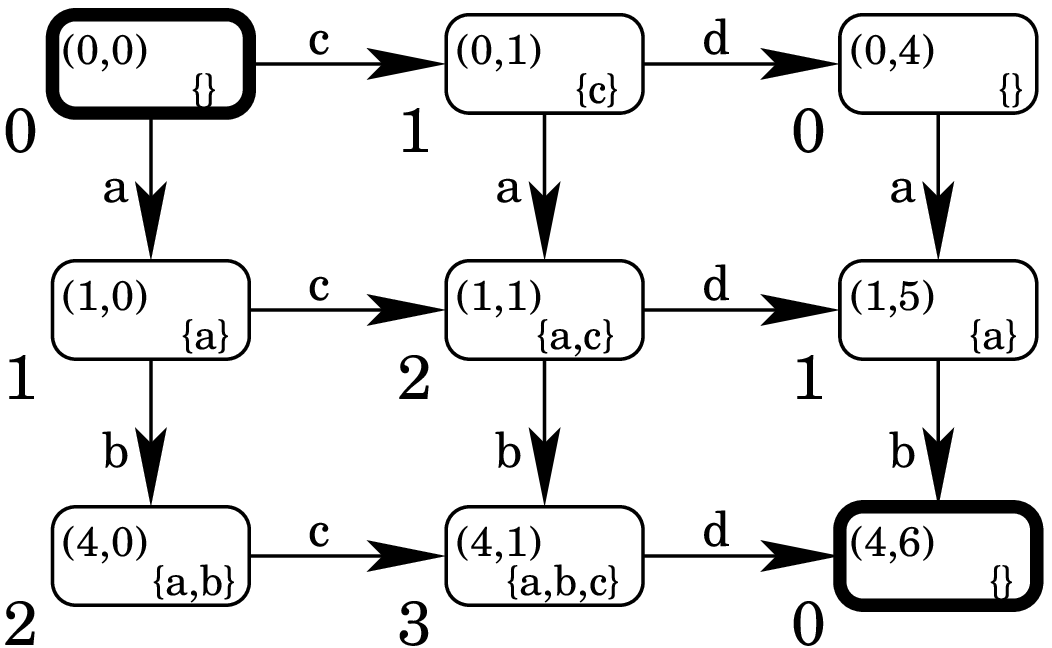}}
\caption{The decision graph of Example \ref{EX1}. It consists of the shaded
vertices from the processing graph in Figure \ref{F07}.}
\label{F08}
\end{figure}

Every decision
graph is directed and acyclic.
Each pallet solution leads to a path from the initial configuration $(0,0, \ldots, 0)$
to the final configuration $(|q_1|, |q_2|, \ldots, |q_k|)$ in the decision graph.
We are interested in such paths where the maximal vertex label on
that path is minimal.

The decision graph can be computed by some breadth first search
algorithm as follows. 
We store the already discovered decision
configurations, i.e.\ the vertices of the graph, in some list $L$.
Initially, list $L$ contains only the initial configuration.
In each step of the algorithm we take the first configuration out
of list $L$. Let $C_Q = (i_1, \ldots, i_k)$ be such a decision configuration.
For each $j \in [k]$ we remove the bin on position $i_j+1$
from sequence $q_j$ and execute all automatic transformation steps.
By this way we reach another decision configuration
$C'_Q = (i'_1, \ldots, i'_k)$ or the final configuration,
and we put this configuration into list $L$, we add
$C'_Q$ to vertex set $V$, and we add $(C_Q, C'_Q)$ to arc set $A$, 
if it is not already contained.

We can combine the construction of the decision graph and the path finding
by some breadth first search algorithm {\sc Find Optimal Pallet Solution},
as shown in Figure \ref{fig:algorithm4}.

\begin{figure}[ht]
\hrule
{\strut\footnotesize \bf Algorithm {\sc Find Optimal Pallet Solution}} 
\hrule
\begin{tabbing}
xxx\=xxx\=xxx\=xxx \= xxx \= xxx \= xxx \= xxx\= xxx \= xxx \= xxx \= xxx\= xxx\=\kill
%xx\=xx\=xx\=xx \= xx \= xx \= xx \= xx\= xx \= xx \= xx \= xx\= xx\=\kill
open$[~(0,\ldots,0)~]:=\emptyset$ \\
\#open$[~(0,\ldots,0)~]:=0$ \\
$val[~(0,\ldots,0)~] := 0$ \\
$L := (~(0,\ldots,0)~)$\> \>\>\>\>\> \>\>\>\>\> $\vartriangleright$ List of configurations to expand\\
$pred[~(0,\ldots,0)~] := \emptyset$ \>\>\>\>\>\>\>\>\> \>\> $\vartriangleright$ List of predecessors of some decision configuration \\
while $L$ is not empty do \\
\> $C := head(L)$\>\>\>\>\>\>\>\>\>\> $\vartriangleright$ let $C=(i_1,\ldots,i_k)$   \\
\>if ($pred[C]\neq \emptyset$) \>\>\>\>\>\>\>\>\>\>$\vartriangleright$ all predecessors of $C$ are computed   \\
\>\> {\sc Extend Path}($C$) \\
\> for $j := 1$ to $k$ do \\
\>\>   $C_s := (i_1, \ldots, i_j+1, \ldots, i_k)$ \\
\>\>  $O:= \OPEN[C]\cup \PL(b_{i_j+1})$ \\
\>\> for $\ell := 1$ to $k$ do \>\>\>\>\>\>\>\>\> $\vartriangleright$ perform automatic transformation steps\\
\>\>\> while the first bin in $q_\ell$ is destinated for an open pallet in $O$ \\
\>\>\>\> let $C_s$ be the configuration obtained by removing the first bin of $q_\ell$ \\
%\>\> while automatic transformation steps are possible \\
%\>\>\> set $C_s$ to the configuration reached by automatic transformation step \\
\>\> if ($C_s$ is not in $L$) \\
\>\>\> compute $\# \OPEN[C_s]$  and $\OPEN[C_s]$ according to Remark \ref{remark-open}\\
\>\>\> $append(C_s,L)$ \\
\>\>   $append(C,pred[C_s])$ 
\end{tabbing}
\hrule
\caption{Construction of the decision graph and computation of an optimal pallet
 solution at once by breadth first search.}
\label{fig:algorithm4}
\end{figure}

The running time of algorithm {\sc Find Optimal Pallet Solution}
can be estimated as follows. 
In each sequence $q_j$ there are bins for at most $m$ pallets. Each pallet
can be opened at most once, so in a decision configuration $(i_1, \ldots, i_k)$
it must hold $i_j+1=\FI(q_j,t)$ for some pallet $t$. And since $i_j=0$ and
$i_j=|q_j|$ are possible as well, the decision graph has at most $(m+2)^k$
vertices each representing a decision configuration in list $L$. 

We use a boolean array of size  $(m+2)^k$ to indicate whether some
decision configuration is already contained in $L$ or not. 
The initialization can be performed in $\bigo((m+2)^k)$. 
Since there are $\bigo((m+2)^k)$ decision configurations each having
at most $k$ predecessors we need at most $\bigo(k\cdot (m+2)^k)$
tests whether some
decision configuration is already contained in $L$. Since
each test can be done in time  $\bigo(k)$, we can bound the running time 
for all tests by  $\bigo(k^2\cdot (m+2)^k)
\subseteq\bigo(n^2\cdot (m+2)^k)$.

The automatic transformation steps can be performed for every $j \in [k]$  
in time $\bigo(n+k)$ and thus for all decision configurations in $L$ in time
$\bigo((m+2)^k \cdot k\cdot (n+k))$.
The computation of $\# \OPEN(C_s)$ for one decision configuration $C_s$ 
can be done in time $\bigo(m\cdot k)$ by Remark \ref{remark-open}. 
Since we only compute $\# \OPEN(C_s)$ for decision configuration $C_s$ 
which is not already contained in $L$, for all decision configurations
we need time in $\bigo((m+2)^k \cdot m \cdot k))$.

Thus the running time of algorithm {\sc Find Optimal Pallet Solution} is in
$\bigo((m+2)^k \cdot k\cdot (n+k) + (m+2)^k \cdot m \cdot k)) \subseteq
\bigo((m+2)^k\cdot n^2)$.

If we have a pallet solution $T = (t_1, \ldots, t_m)$ computed by any FIFO stack-up algorithm,
we can convert the pallet solution into a sequence of transformation
steps, i.e.\ a processing of $Q$, in time $\bigo(n\cdot k)\subseteq\bigo(n^2)$ by
some simple algorithm: Repeatedly, in the $i$-th decision
configuration choose the pallet $t_i$, remove a bin for pallet $t_i$,
and execute all automatic transformation steps, until the next decision
configuration is reached. If no bin for pallet $t_i$ can be removed in the $i$-th decision
configuration, or if more than $p$ pallets are opened, reject the input. 
Each transformation step 
of these at most $n$ steps can be done in time $\bigo(k)$ by Equation
(\ref{transStep}).

\begin{theorem}\label{th-main}
The {\sc FIFO Stack-Up} problem 
can be solved in time $\bigo(n^2 \cdot (m+2)^k)$.
\end{theorem}

Another way to describe a path from the initial configuration
to the final configuration in the decision graph is as follows.
Let $s_i\in [k]$ denote the sequence from which  a bin for pallet
$t_i$ will be removed in the next transformation step 
after reading the $i$-th decision configuration on a path. 
Then $S=(s_1, \ldots, s_m)$ is called a {\em sequence
solution}. In Example \ref{EX1} we have $S=(2,2,1,1)$ as a sequence solution.

There are at most $k^m$ sequence orders $(s_1,\ldots,s_m)$, where
$s_i \in [k]$. The following algorithm checks in time
$\bigo(n \cdot k) \subseteq \bigo(n^2)$, whether one
sequence order describes a sequence
solution: In the $i$-th decision configuration remove the bin from sequence $s_i$.
Reject the input, if this is impossible. Otherwise 
do all automatic transformation steps until the next decision configuration
is reached. Each step can be done in time $\bigo(k)$ by Equation
(\ref{transStep}). Reject the input, if more than $p$ pallets are opened.

If we enumerate all $k^m$ possible sequence orders and check each order, 
we can solve the {\sc FIFO Stack-Up} problem.

\begin{theorem}\label{th-main-a}
The {\sc FIFO Stack-Up} problem 
can be solved in time $\bigo(n^2 \cdot k^m)$.
\end{theorem}

%%%%%%%%%%%%%%%%%%%%%%%%%%%%%%%%%%%%%%%%%%%%%%%%%%%%%%%%%%%%%%%%%%%%%%%%%%
%%%%%%%%%%%%%%%%%%%%%%%%%%%%%%%%%%%%%%%%%%%%%%%%%%%%%%%%%%%%%%%%%%%%%%%%%%
\section{Application of Methods to Solve Hard Problems to the  {\sc FIFO Stack-Up} Problem}\label{sec-algo}
%%%%%%%%%%%%%%%%%%%%%%%%%%%%%%%%%%%%%%%%%%%%%%%%%%%%%%%%%%%%%%%%%%%%%%%%%%
%%%%%%%%%%%%%%%%%%%%%%%%%%%%%%%%%%%%%%%%%%%%%%%%%%%%%%%%%%%%%%%%%%%%%%%%%%

Since the {\sc FIFO Stack-Up} problem is NP-complete in general,
it is unlikely
that there exist polynomial time algorithms for the problem. 
This motivates us to consider 
restricted versions of the problem,
linear programming solutions,
exponential time algorithms, 
parameterized algorithms, and
approximation algorithms.

%%%%%%%%%%%%%%%%%%%%%%%%%%%%%%%%%%%%%%%%%%%%%%%%%%%%%%%%%%%%%%%%%%%%%%%%%%%%%%%%
\subsection{Restricted Versions}
%%%%%%%%%%%%%%%%%%%%%%%%%%%%%%%%%%%%%%%%%%%%%%%%%%%%%%%%%%%%%%%%%%%%%%%%%%%%%%%%

We have shown that the {\sc FIFO Stack-Up} problem can be solved
in polynomial time for the following special cases.

\begin{itemize}
\item If the number $k$ of sequences is constant then the {\sc FIFO Stack-Up}
  problem can be solved in polynomial time $\bigo(k \cdot (N+1)^k)$, see Theorem
  \ref{th-processing-gr}.

\item If the number $p$ of stack-up places is constant we get a polynomial
  running time of $\bigo(n + m^{p+2})$ by constructing the sequence graph,
  see Section \ref{sec-xp}.

\item If the number $m$ of pallets is constant we get an algorithm with polynomial
  running time $\bigo(n^2 \cdot m!)$, see Section \ref{sec-fpt}.

\item Finally, if the number $n$ of bins is constant, then we restrict the input
  size, see Equation (\ref{eq-inst}) and Remark \ref{remark-practical-br}, 
and therefore we get a constant running time $\bigo(1)$.
\end{itemize}

%%%%%%%%%%%%%%%%%%%%%%%%%%%%%%%%%%%%%%%%%%%%%%%%%%%%%%%%%%%%%%%%%%%%%%%%%%%%%%%%
\subsection{Linear Programming}
%%%%%%%%%%%%%%%%%%%%%%%%%%%%%%%%%%%%%%%%%%%%%%%%%%%%%%%%%%%%%%%%%%%%%%%%%%%%%%%%

%
% stimmt das n^2 ???? durchs linearisieren mehr???
%

Next we give two linear programming approaches to solve the optimization version of the 
{\sc FIFO Stack-Up} problem in which we have to compute a minimum number $p$ such
that a given list $Q$ of sequences can be processed with at most $p$ stack-up places.
The first one computes a bin solution using $\bigo(n^2)$ 
variables and the second
one computes a pallet solution using $\bigo(m^4)$ variables. 
In Section \ref{sec-exp} we will compare both solutions 
implemented in CPLEX and GLPK
within an experimental study of running times.

%%%%%%%%%%%%%%%%%%%%%%%%%%%%%%%%%%%%%%%%%%%%%%%%%%%%%%%%%%%%%%%%%%%%%%%%%%%%%%%%
\subsubsection{Computing a Bin Solution}\label{sec-IPbin}
%%%%%%%%%%%%%%%%%%%%%%%%%%%%%%%%%%%%%%%%%%%%%%%%%%%%%%%%%%%%%%%%%%%%%%%%%%%%%%%%

We have given a list $Q$ of $k$ sequences and $n$ bins $b_1,\ldots,b_{n}$ 
which use $m$ pallet symbols. Our aim is to find a bin solution for $Q$,
i.e.\ a bijection $\pi:[n]\to [n]$ for the bins, such that by the removal
of the bins from the sequences in the order of $\pi$  the number 
of needed stack-up places $p$ is minimized. 

To realize a bijection $\pi$ we define $n^2$ binary variables
$x_i^j\in\{0,1\}$, $i,j\in [n]$,
such that $x_i^j$ is equal to 1, if and only if bin $b_i$ is placed at 
position $j$ by $\pi$.
In order to map every bin $b_i$, $i\in [n]$, on exactly one position,
i.e.\ to ensure $\pi$ to be surjective,
we use the conditions
$$\sum_{j=1}^{n} x_i^j  = 1  \mbox{ for every } i \in [n]$$  
and in order to map on every position $j\in [n]$ exactly one
bin, i.e.\ to ensure $\pi$ to be injective, we use the conditions
$$\sum_{i=1}^{n} x_i^j  = 1  \mbox{ for every } j\in [n].$$   
Further we have to ensure that all variables $x_i^j$, 
$i,j\in [n]$, are in $\{0,1\}$.  We will denote the previous
$n^2+2n$  conditions by $\text{\sc Permutation}(n, x_i^j)$.

By $\pi$ the relative order of the bins of each sequence $q \in Q$ has to be
preserved, consider for example the bin solution of the sequences of Example \ref{EX1}. The bins
of the corresponding sequence are shown in black, the others are greyed out.
\[
\begin{array}{l}
  ({\gr b_5}, {\gr b_6}, {\gr b_7}, {\gr b_8}, b_1, {\gr b_9}, b_2, {\gr b_{10}},
    b_3, b_4) ~~ \rightarrow ~~ q_1 \\[1ex]
  (b_5, b_6, b_7, b_8, {\gr b_1}, b_9, {\gr b_2}, b_{10}, {\gr b_3}, {\gr b_4})
    ~~ \rightarrow ~~ q_2
\end{array}
\]
We have to consider the orderings of the bins given by the $k$ sequences
$q_1=(b_1, \ldots, b_{n_1}), \ldots, q_k=(b_{n_{k-1}+1}, \ldots, b_{n_k})$.
For every sequence $q_\ell$, $1\leq \ell \leq k$, and
every bin $b_i$ of this sequence, i.e. $n_{\ell-1}+1 \leq i \leq n_\ell$, 
we know that if $b_i$ is placed at position $j$, i.e. $x_{i}^{j}=1$,
then 
\begin{itemize}
\item every bin $b_{i'}$, $i'>i$ of sequence $q_\ell$, i.e.
  $i < i' \leq n_\ell$, is not placed before $b_i$, i.e. $x_{i'}^{j'}=0$
  for all $j'<j$, which is ensured by $x_{i'}^{j'}\leq 1-x_{i}^{j}$ and
\item every bin $b_{i'}$, $i'<i$ of sequence $q_\ell$, i.e.
  $n_{\ell-1}+1 \leq i'< i$, is not placed after $b_i$, i.e.
  $x_{i'}^{j'}=0$ for all $j'>j$, which is ensured by $x_{i'}^{j'}\leq 1-x_{i}^{j}$.
\end{itemize}
Since we have $\bigo(n^2)$ pairs $(j',j)$ and  $\bigo(N^2)$
pairs $(i',i)$, and since we have $k$ sequences, there are at most
$\bigo(k \cdot n^2 \cdot N^2)$ such
conditions.  We will denote these
conditions by $\text{\sc SequenceOrder}(Q, x_i^j)$.

By an integer valued variable $p$  we count the number of used stack-up places 
for some given sequence $Q$ as follows. 
\begin{eqnarray}
 \text{minimize } p \label{c1} 
\end{eqnarray}
subject to
\begin{eqnarray}
 & & \text{\sc Permutation}(n, x_i^j), \text{ and }
     \text{\sc SequenceOrder}(Q, x_i^j) \label{c2} \\
 & \text{ and }
   & \sum_{t=1}^{m} f(t,c) \leq p \text{ for every }c\in [n-1]  \label{c3}  \\
 & \text{ and }
   & f(t,c) = \bigg(\bigvee_{i\in[n],j\leq c, \PL(b_i)=t} x_i^j\bigg) \wedge
     \bigg(\bigvee_{i\in[n],j > c,\PL(b_i)=t}  x_i^j\bigg) \label{c4}
\end{eqnarray}

For the correctness note that subexpression $f(t,c)$ is equal to one 
if and only if in the considered ordering of the bins there is
a bin $b'$ with $\PL(b')=t$ opened at a step $\leq c$ and 
there is
a bin $b''$ with $\PL(b'')=t$ opened at a step $>c$, 
if and only if 
pallet $t$ is open after the $c$-th bin has been removed.

\begin{remark}\label{remark-linearisierelp1}
Equation (\ref{c4}) is propositional logic and not a linear function in standard
form. Propositional logic can be reformulated for binary integer linear
programming using the results of \cite{Gur14} which show that 
every $n$-ary boolean function  $f(x_1,\ldots,x_n)=x_{n+1}$ 
can be defined with a binary linear program using $n+1$ boolean variables 
$x_1,\ldots,x_{n+1}$.  In order to express Equation (\ref{c4}) we need
to define binary conjunctions and an $n$-ary disjunction.
\begin{itemize}
\item Every conjunction $x_1\wedge x_2$, $x_1,x_2\in\{0,1\}$ can be realized by 
introducing a new variable $x_3\in\{0,1\}$ and three conditions
%$$x_3\leq x_1, ~~ x_3 \leq x_2, ~~ x_1+x_2\leq x_3+1$$ 
$$x_3- x_1\leq 0, ~~ x_3 - x_2\leq 0, ~~ x_1+x_2-x_3 \leq 1$$ 
such that finally $x_3=x_1\wedge x_2$.

\item Every $n$-ary disjunction $x_1\vee x_2 \vee \ldots \vee x_n$ can be realized by 
introducing a new variable $x_{n+1}\in\{0,1\}$ and $n+1$ conditions
%$$x_1\leq x_{n+1}, ~~ x_2 \leq x_{n+1}, ~~ \ldots, ~~ x_n\leq x_{n+1}, ~~ x_1+x_2 +\ldots + x_n \geq x_{n+1}$$ 
$$x_1- x_{n+1}\leq 0, ~~ x_2 - x_{n+1}\leq 0, ~~ \ldots, ~~ x_n - x_{n+1}\leq 0, ~~ x_1+x_2 +\ldots + x_n - x_{n+1}\geq 0$$ 
such that finally $x_{n+1}=x_1\vee x_2 \vee \ldots \vee x_n$.
\end{itemize}

Next these two transformations are applied on Equation (\ref{c4}).
We define $m\cdot(n-1)\leq n^2$ binary variables $g(t,c)\in\{0,1\}$, $t\in[m]$, $c\in[n-1]$ such that
$$g(t,c)=\bigvee_{i\in[n],j\leq c, \PL(b_i)=t} x_i^j,$$ which can be realized by 
\begin{eqnarray}
   &              & x_i^j -g(t,c)\leq 0 ~~~ \text{ for every } i\in[n],j\leq c, \PL(b_i)=t,      t\in[m], c\in[n-1]      \label{xxz1} \\
   & \text{ and } &  \big(\sum_{i\in[n],j\leq c, \PL(b_i)=t} x_i^j\big) - g(t,c)\geq 0 ~~~ \text{ for every } t\in[m], c\in[n-1] \label{xxz2}
\end{eqnarray}

Further we define  $m\cdot (n-1) \leq n^2$ binary variables $h(t,c)\in\{0,1\}$, $t\in[m]$, $c\in[n-1]$ such that
$$h(t,c)=\bigvee_{i\in[n],j > c,\PL(b_i)=t}  x_i^j,$$ which can be realized by 
\begin{eqnarray}
   &              & x_i^j - h(t,c) \leq 0~~~ \text{ for every } i\in[n],j > c,\PL(b_i)=t,    t\in[m], c\in[n-1]      \label{xxz3} \\
   & \text{ and } &  \big(\sum_{i\in[n],j > c,\PL(b_i)=t} x_i^j\big) - h(t,c)\geq 0 ~~~ \text{ for every } t\in[m], c\in[n-1] \label{xxz4}
\end{eqnarray}

Finally we define  $m\cdot (n-1) \leq n^2$ binary variables $f(t,c)\in\{0,1\}$, $t\in[m]$, $c\in[n-1]$, such that 
$f(t,c)= g(t,c) \wedge h(t,c)$, which
can be realized  by
\begin{eqnarray}
   &              &f(t,c) - g(t,c) \leq 0~~~ \text{ for every } t\in[m], c\in[n-1] \label{xxz5} \\
   & \text{ and } &f(t,c)  - h(t,c) \leq  0~~~  \text{ for every }  t\in[m], c\in[n-1] \label{xxz6}\\
   & \text{ and } & g(t,c) + h(t,c)- f(t,c)\leq 1  ~~~ \text{ for every } t\in[m], c\in[n-1]\label{xxz7}
\end{eqnarray}
\end{remark}

\begin{theorem}\label{th-lp-n}
For every list $Q$ of sequences the integer linear program
(\ref{c1})-(\ref{c3}),(\ref{xxz1})-(\ref{xxz7}) computes the minimum number of
stack-up places $p$ in a processing of $Q$.
\end{theorem}

%%%%%%%%%%%%%%%%%%%%%%%%%%%%%%%%%%%%%%%%%%%%%%%%%%%%%%%%%%%%%%%%%%%%%%%%%%%%%%%%
\subsubsection{Computing a Pallet Solution}\label{sec-IPplt}
%%%%%%%%%%%%%%%%%%%%%%%%%%%%%%%%%%%%%%%%%%%%%%%%%%%%%%%%%%%%%%%%%%%%%%%%%%%%%%%%

By Theorem \ref{th-pw-gbefore} the minimum number of
stack-up places can be computed by the directed pathwidth of 
the sequence graph $G_Q$ plus one.
In the following we use the fact that the
directed pathwidth equals the directed vertex separation number
\cite{YC08}.

For a digraph $G=(V,A)$ on $n$ vertices, we denote by $\Pi(G)$ the set of all 
bijections $\pi:[n]\to [n]$ of its vertex set.
Given a bijection $\pi\in \Pi(G)$ we define for $i \in [n]$ the vertex sets 
$L(i,\pi,G)=\{u\in V ~|~ \pi(u)\leq i \}$ and 
$R(i,\pi,G)=\{u\in V ~|~ \pi(u) > i \}$.
Every position $i \in [n]$ is called a cut. This allows us to define
the directed vertex separation number for digraph $G$ as follows.
\[
    \dvsns(G) = \min_{\pi\in \Pi(G)} \max_{1\leq i\leq |V|} |
    \{u\in L(i,\pi,G)~|~ \exists v \in R(i,\pi,G): (v,u)\in A\}|
\]

An integer linear program for computing the directed vertex separation number for
some given sequence graph $G_Q=(V,A)$ on $m$ vertices is as follows. To realize
$\pi$ we define $m^2$ binary variables $x_i^j\in\{0,1\}$, $i,j\in [m]$,
such that $x_i^j$ is equal to 1, if and only if pallet $v_i$ is placed at 
position $j$ by $\pi$.
Additionally we use a variable $w$ in order to count the vertices 
adjacent to the right side of the cuts.
\begin{eqnarray}
\text{minimize } w \label{f-path}
\end{eqnarray}
subject to
\begin{eqnarray}
 & & \text{\sc Permutation}(m, x_i^j) \label{c1-path} \\
 & \text{ and } & \sum_{j=1}^{c} Y(j,c) \leq w \text{ for every } c\in [m-1] \label{c2-path} \\
 & \text{ and }
   & Y(j,c) = \bigvee_{\substack{j'\in\{c+1,\ldots,m\}\\ i,i'\in [m], (v_{i'},v_{i})\in A}} (x_{i}^{j} \wedge x_{i'}^{j'})    \label{c3-path}  
\end{eqnarray}

For the correctness note that subexpression $Y(j,c)$ is equal 
to one if and only if there exists 
an arc from a vertex on a position $j'>c$ to a vertex on position $j$.

\begin{remark}\label{remark-linearisierelp2}
Equation (\ref{c3-path}) is propositional logic and not a linear function in standard
form. In order to express Equation (\ref{c3-path}) we need
to define binary conjunctions and an $n$-ary disjunction, see Remark \ref{remark-linearisierelp1}.

We define $m^4$ binary variables $X(i,i',j,j')\in\{0,1\}$, $i,i',j,j'\in[m]$, such that
$X(i,i',j,j')=x_{i}^{j} \wedge x_{i'}^{j'}$, which can be realized by 
\begin{eqnarray}
   &              & X(i,i',j,j') - x_{i}^{j} \leq 0 ~~~ \text{ for every }  i,i',j,j'\in[m] \label{x1} \\
   & \text{ and } & X(i,i',j,j') - x_{i'}^{j'} \leq 0~~~  \text{ for every }  i,i',j,j'\in[m] \label{x2}\\
   & \text{ and } & x_{i}^{j} + x_{i'}^{j'}- X(i,i',j,j')\leq1  ~~~ \text{ for every } i,i',j,j'\in[m] \label{x3}
\end{eqnarray}

We define $\bigo(m^2)$ binary variables $Y(j,c)\in\{0,1\}$ for $j,c\in[m-1]$, $j\leq c$, such that 
$$Y(j,c)= \bigvee_{\substack{j'\in\{c+1,\ldots,m\}\\ i,i'\in [m], (v_{i'},v_{i})\in A}} X(i,i',j,j'),$$ which
can be realized  by
\begin{eqnarray}
  &              & \hspace{-2cm} X(i,i',j,j') - Y(j,c) \leq 0~~ \text{for every } j'\in\{c+1,\ldots,m\}, i,i'\in [m], (v_{i'},v_{i})\in A,      j,c\in[m-1] \label{xx1} \\
   & \text{ and } &  \big(\sum_{\substack{j'\in\{c+1,\ldots,m\}\\ i,i'\in [m], (v_{i'},v_{i})\in A}} X(i,i',j,j')\big) - Y(j,c) \geq 0~~~ \text{ for every } j,c\in[m-1]\label{xx2}
\end{eqnarray}
\end{remark}

\begin{theorem}\label{th-lp-m}
For every list $Q$ of sequences the integer linear program
(\ref{f-path})-(\ref{c2-path}),(\ref{x1})-(\ref{xx2}) computes the minimum number
of stack-up places $p = w+1$ in a processing of $Q$.
\end{theorem}

%%%%%%%%%%%%%%%%%%%%%%%%%%%%%%%%%%%%%%%%%%%%%%%%%%%%%%%%%%%%%%%%%%%%%%%%%%%%%%%%
\subsection{Exact Exponential Time Algorithms}
%%%%%%%%%%%%%%%%%%%%%%%%%%%%%%%%%%%%%%%%%%%%%%%%%%%%%%%%%%%%%%%%%%%%%%%%%%%%%%%%

The running time of exponential time algorithms is often given in
the $\bigo^*$-notation, which hides polynomial factors.

\begin{theorem}
The {\sc FIFO Stack-Up} problem can be solved in time $\bigo^*(2^m)$.
\end{theorem}

\begin{proof}
The directed pathwidth of a digraph $G=(V,A)$ can be computed in
time $\bigo^*(2^{|V|})$ by \cite{BFKKT12}. Since the sequence graph
$G_Q$ can be constructed in time $\bigo(n+ k\cdot m^2)$  by  algorithm
{\sc Create Sequence Graph} shown in Figure \ref{fig:algorithm5}, the statement follows
by Theorem \ref{th-pw-gbefore}. \qed
\end{proof}

%%%%%%%%%%%%%%%%%%%%%%%%%%%%%%%%%%%%%%%%%%%%%%%%%%%%%%%%%%%%%%%%%%%%%%%%%%%
\subsection{Parameterized Algorithms}
%%%%%%%%%%%%%%%%%%%%%%%%%%%%%%%%%%%%%%%%%%%%%%%%%%%%%%%%%%%%%%%%%%%%%%%%%%%

Within parameterized complexity we consider a two dimensional 
analysis of the computational complexity of a problem. Denoting the
input by $I$, the two considered dimensions are the size $|I|$ 
and a parameter $\kappa(I)$,  see \cite{FG06} for a survey.

A {\em parameterized problem} is a pair $(\Pi,\kappa)$, where $\Pi$
is a decision problem, ${\mathcal I}$ the set of all instances
of $\Pi$ and $\kappa: \mathcal{I} \to \IN$ is
a so called {\em parameterization} or {\em parameter}.
The parameter $\kappa(I)$ should be
small for all inputs $I \in {\mathcal I}$.
Please note that the following running times are
exponential,
but for small values of $\kappa(I)$ they may be good in practice.

\begin{itemize}
\item An algorithm $A$ is an {\em xp-algorithm with respect to $\kappa$}, if
  there are two computable functions $f,g: \IN \to \IN$ 
  such that for every instance $I \in {\mathcal I}$ the running time
  of $A$ on $I$ (with input size $|I|$) is at most  
  $f(\kappa(I))\cdot |I|^{g(\kappa(I))}$. A typical running time is
  $2^{\kappa(I)} \cdot |I|^{3 \cdot \kappa(I)}$.
  A parameterized problem  $(\Pi,\kappa)$ belongs to the class $\xp$
  and is called {\em slicewise polynomial},
  if there is an xp-algorithm with respect to $\kappa$ which decides $\Pi$.
\item An algorithm $A$ is an {\em fpt-algorithm with respect to $\kappa$},
  if there is a computable function $f: \IN \to \IN$ 
  such that for every instance $I\in {\mathcal I}$ the running time
  of $A$ on $I$ (with input size $|I|$) is at most  
  $f(\kappa(I))\cdot |I|^c$ for some fixed $c \in \IN$. A typical running time is
  $2^{\kappa(I)} \cdot |I|^2$. 
  A parameterized problem  $(\Pi,\kappa)$  
  belongs to the class $\fpt$ and is called {\em fixed-parameter tractable},
  if there is an fpt-algorithm with respect to $\kappa$ which decides $\Pi$.
\end{itemize}

Fxed-parameter algorithms have shown to be useful in
several fields, among there are: phylogenetics \cite{GNT08}, 
biopolymer sequences comparison \cite{CHLRS08},
artificial intelligence, constraint satisfaction, and database problems \cite{GS08},
geometric problems \cite{GKW08}, and
cognitive modeling \cite{RW08}.

In order to show fixed-parameter intractability, it is useful to
show the hardness with respect to one of the classes $\w[t]$ for
some $t\geq 1$ which were introduced by Downey and Fellows \cite{DF99}
in terms of weighted satisfiability problems on classes of circuits.
The following relations -- the so called $\w$-hierarchy -- hold
and all inclusions are assumed to be strict.
$$\fpt \subseteq \w[1] \subseteq \w[2] \subseteq \ldots \subseteq \xp$$

For the {\sc FIFO Stack-up} problem we choose the number
of sequences $k$ or the number of pallets $m$ and others as a parameter $\kappa(I)$ from the instance
$I$, in order to obtain the following parameterized problem.

\begin{desctight}
\item[Name] $\kappa(I)$-{\sc FIFO Stack-up}

\item[Instance] A list $Q = (q_1, \ldots, q_k)$ of sequences and a positive
integer $p$.

\item[Parameter] $\kappa(I)$

\item[Question] Is there a processing of $Q$, such that in each configuration during
the processing of $Q$ at most $p$ pallets are open?
\end{desctight}

We give two useful instruments to obtain fpt-results used in the next subsections.
The first result gives a connection between the existence of ILP formulations 
and fixed-parameter tractability w.r.t.\ the number of variables. 

\begin{theorem}[\cite{Kan87},\cite{Len83}]\label{th-lenstra}
If for some problem $\Pi$ there is an ILP using $\ell$ variables, then $\Pi$ can
be solved for every instance $I$ in time $\bigo(|I|\cdot \ell^{\bigo(\ell)})$.
\end{theorem}

The next remark states the running time of an exhaustive search 
for problems in NP. The correctness follows immediately by the definition
of NP-completeness.

\begin{remark}[\cite{Woe04}]\label{rem-woe}
If some problem $\Pi$ belongs to NP, then there is some polynomial $q$ such that  $\Pi$ can
be solved for every instance $I$ in time $\bigo(2^{q(|I|)})$.
\end{remark}

%%%%%%%%%%%%%%%%%%%%%%%%%%%%%%%%%%%%%%%%%%%%%%%%%%%%%%%%%%%%%%%%%%%%%%%%%%%%%%%%
\subsubsection{XP-Algorithms} \label{sec-xp}
%%%%%%%%%%%%%%%%%%%%%%%%%%%%%%%%%%%%%%%%%%%%%%%%%%%%%%%%%%%%%%%%%%%%%%%%%%%%%%%%

\paragraph{Parameterization by number of sequences $k$}
By Theorem \ref{th-processing-gr} we obtain an 
xp-algorithm with respect to the parameter $k$.

\begin{theorem}
There is an xp-algorithm that solves $k$-{\sc FIFO Stack-Up}
in time $\bigo(k \cdot (N+1)^k)$.
\end{theorem}

%%%%%%%%%%%%%%%%%%%%%%%%%%%%%%%%%%%%%%%%%%%%%%%%%%%%%%%%%%%%%%%%%%%%%%%%%%%%%%%%%

\paragraph{Parameterization by number of stack-up places $p$}
The directed pathwidth of a digraph $G=(V,A)$ can be computed in time 
$\bigo(|A|\cdot |V|^{\dpw(G)+1})\subseteq \bigo(|V|^{\dpw(G)+3})$ by \cite{Tam11}
and the sequence graph $G_Q$ can be computed in time $\bigo(n+k\cdot m^2)$ which
can be bounded by $\bigo(n+ m^3)$ due to Remark \ref{remark-practical-br}.
Thus the  {\sc FIFO Stack-Up} problem can be solved
by Theorem \ref{th-pw-gbefore} in time $\bigo(n+m^{p+2})$.

\begin{theorem}
There is an xp-algorithm that solves $p$-{\sc FIFO Stack-Up}
in time $\bigo(n+m^{p+2})$.
\end{theorem}

%%%%%%%%%%%%%%%%%%%%%%%%%%%%%%%%%%%%%%%%%%%%%%%%%%%%%%%%%%%%%%%%%%%%%%%%%%%%%%%%
\subsubsection{FPT-Algorithms} \label{sec-fpt}
%%%%%%%%%%%%%%%%%%%%%%%%%%%%%%%%%%%%%%%%%%%%%%%%%%%%%%%%%%%%%%%%%%%%%%%%%%%%%%%%

\paragraph{Parameterization by number of bins $n$}
We generate all possible bin orders $(b_{\pi(1)},\ldots,b_{\pi(n)})$
and verify, whether this can be a processing. This leads to a simple but
very inefficient algorithm with running time $\bigo(n^2\cdot n!)$,
where $\bigo(n^2)$ time is needed for each verification.

\begin{theorem}
There is an fpt-algorithm that solves $n$-{\sc FIFO Stack-Up}
in time $\bigo(n!\cdot n^2)$.
\end{theorem}

Alternatively we can apply Theorem \ref{th-lenstra}, Remark \ref{remark-linearisierelp1}, and Theorem \ref{th-lp-n}.
Integer program (\ref{c1})-(\ref{c3}),(\ref{xxz1})-(\ref{xxz7})  has at most $4n^2+1\in \bigo(n^2)$ variables and a polynomial
number of constraints.
From Theorem \ref{th-lenstra} it follows that the {\sc FIFO Stack-Up}
problem is fixed-parameter tractable for the parameter $n$.

\begin{theorem}
There is an fpt-algorithm that solves $n$-{\sc FIFO Stack-Up}
in time $\bigo(|I| \cdot (4n^2+1)^{\bigo(n^2)})$.
\end{theorem}

Alternatively we can apply Remark \ref{rem-woe}.
Since we assume that $p\leq m \leq n$, the size of an instance $I$
can be bounded by $|I|\in\bigo(n\cdot\log(n))$, see Equation (\ref{eq-inst}), 
and by Remark \ref{rem-woe}
it follows that the {\sc FIFO Stack-Up}
problem is fixed-parameter tractable for the parameter $n$.

\begin{theorem}
There is an fpt-algorithm that solves $n$-{\sc FIFO Stack-Up}
in time $\bigo(2^{q(n\cdot \log(n))})$ for some polynomial $q$.
\end{theorem}

%%%%%%%%%%%%%%%%%%%%%%%%%%%%%%%%%%%%%%%%%%%%%%%
\paragraph{Parameterization by number of pallets $m$}
Let $(t_1, \ldots, t_m)$ be a permutation of the pallets of $\PLS(Q)$. It
can be checked in time $\bigo(n \cdot k) \subseteq \bigo(n^2)$ whether this
permutation describes a processing of $Q$ using at most $p$ stack-up places,
see Section \ref{SCdec}.
There are $m!$ permutations of the pallets. If we enumerate all $m!$ possible pallet
orders and check each order, we can solve the  {\sc FIFO Stack-Up} problem.

\begin{theorem}
There is an fpt-algorithm that solves $m$-{\sc FIFO Stack-Up}
in time $\bigo(n^2 \cdot m!)$.
\end{theorem}

Alternatively we can apply Theorem \ref{th-lenstra}, Remark \ref{remark-linearisierelp2}, and Theorem \ref{th-lp-m}.
Integer linear program (\ref{f-path})-(\ref{c2-path}),(\ref{x1})-(\ref{xx2}) has at most $m^4+2m^2+1\in\bigo(m^4)$ variables and 
a polynomial number of constraints.
From Theorem \ref{th-lenstra} it follows that the {\sc FIFO Stack-Up}
problem is fixed-parameter tractable for the parameter $m$.

\begin{theorem}
There is an fpt-algorithm that solves $m$-{\sc FIFO Stack-Up}
in time $\bigo(|I| \cdot (m^4+2m^2+1)^{\bigo(m^4)})$.
\end{theorem}

We also can apply the relation to directed pathwidth given in Theorem \ref{th-pw-gbefore}.
Since the directed pathwidth of a digraph $G=(V,A)$ can be computed in
time $\bigo(1.89^{|V|})$ by \cite{KKKTT12}, and since the sequence graph
$G_Q$ can be constructed in time $\bigo(n+k\cdot m^2)$  by  algorithm
{\sc Create Sequence Graph} shown in Figure \ref{fig:algorithm5}, the next result follows
by Theorem \ref{th-pw-gbefore}.

\begin{theorem}
There is an fpt-algorithm that solves $m$-{\sc FIFO Stack-Up}
in time 
$\bigo(n+k\cdot m^2 + 1.89^m)$.
\end{theorem}

%
%neue idee von jochen, noch nicht in der submited version zu Cor:
%

Alternatively we can apply the relation to directed pathwidth given in 
Theorem \ref{th-pw-gbefore} and Proposition \ref{prop}.
Let $(Q,p)$ be an instance for the {\sc FIFO Stack-Up} problem with values
$(n,m,k,N)$. We start to compute from $Q$ the sequence graph $G_Q=(V,A)$ 
in time $\bigo(n+k\cdot m^2)$. Then from $G_Q$ we build 
up the sequence system $Q'_{G_Q}$  with values $(n',m',k',N')$ in time $\bigo(|V|+|A|)\subseteq\bigo(m^2)$.
By Theorem \ref{th-pw-gbefore} and Proposition \ref{prop} list $Q$ can be processed
with at most  $p$  stack-up places, if and only if list $Q'$ can be processed
with at most  $p$  stack-up places. Further we know that $k'\in\bigo(m^2)$ and
$N'=2$. By applying Theorem  \ref{th-processing-gr} for $Q'$ 
we can solve the {\sc FIFO Stack-Up} problem
in time $\bigo(k'\cdot (N'+1)^{k'})\subseteq \bigo(m^2 \cdot 3^{m^2})$.

\begin{theorem}
There is an fpt-algorithm that solves $m$-{\sc FIFO Stack-Up}
in time 
$\bigo(m^2 \cdot 3^{m^2}+n+k\cdot m^2)$.
\end{theorem}

\paragraph{Parameterization by combined parameters}

In the case of $\w[1]$-hardness with respect to some parameter $\ell$ a natural
question is whether the problem remains hard for {\em combined} parameters, i.e.
parameters $(\ell,\ell')$ that consists of two or even more parts of the input. 
Since the  existence of an fpt-algorithm w.r.t.\ parameter $k$ is open up to now,
we next conclude an fpt-algorithm with respect to parameter $(k,m)$,
by the result of Theorem \ref{th-main-a}.

\begin{theorem}
There is an fpt-algorithm that solves $(k,m)$-{\sc FIFO Stack-Up}
 in time $\bigo(n^2 \cdot k^m)$.
\end{theorem}

In practice $k$ is much smaller than $m$, since there are much 
fewer buffer conveyors than pallets, thus this solution is better 
than $\bigo(n^2 \cdot m!)$.
In order to show a better fpt-algorithm with respect to parameter $(k,m)$,
we consider the result of Theorem \ref{th-main}.

\begin{theorem}
There is an fpt-algorithm that solves $(k,m)$-{\sc FIFO Stack-Up}
in time $\bigo(n^2 \cdot (m+2)^k)$.
\end{theorem}

Since there are much fewer buffer conveyors than pallets in practice,  this 
solution is better than $\bigo(n^2 \cdot k^m)$.

%%%%%%%%%%%%%%%%%%%%%%%%%%%%%%%%%%%%%%%%%%%%%%%%%%%%%%%%%%%%%%%%%%%%%%%%%%%%%%%%
\subsection{Approximation}
%%%%%%%%%%%%%%%%%%%%%%%%%%%%%%%%%%%%%%%%%%%%%%%%%%%%%%%%%%%%%%%%%%%%%%%%%%%%%%%%

Since the directed pathwidth of a digraph $G=(V,A)$ 
can be approximated by a factor of $\bigo(\log^{1.5}|V|)$ by \cite{KKK15}, the
{\sc FIFO Stack-Up} problem can be approximated using the sequence graph $G_Q$ 
by Theorem \ref{th-pw-gbefore}.

\begin{theorem}
There is an approximation algorithm that solves the optimization version of the 
{\sc FIFO Stack-Up} problem up to a factor of $\bigo(\log^{1.5} m)$.
\end{theorem}

%%%%%%%%%%%%%%%%%%%%%%%%%%%%%%%%%%%%%%%%%%%%%%%%%%%%%%%%%%%%%%%%%%%%%%%%%%%%%%%%
%%%%%%%%%%%%%%%%%%%%%%%%%%%%%%%%%%%%%%%%%%%%%%%%%%%%%%%%%%%%%%%%%%%%%%%%%%%%%%%%
%%%%%%%%%%%%%%%%%%%%%%%%%%%%%%%%%%%%%%%%%%%%%%%%%%%%%%%%%%%%%%%%%%%%%%%%%%%%%%%%
\section{Experimental Results} \label{sec-exp}
%%%%%%%%%%%%%%%%%%%%%%%%%%%%%%%%%%%%%%%%%%%%%%%%%%%%%%%%%%%%%%%%%%%%%%%%%%%%%%%%
%%%%%%%%%%%%%%%%%%%%%%%%%%%%%%%%%%%%%%%%%%%%%%%%%%%%%%%%%%%%%%%%%%%%%%%%%%%%%%%%

Next we want to evaluate an implementation of algorithm 
{\sc Find Optimal Pallet Solution} given in Figure \ref{fig:algorithm4} 
and our two linear programming approaches given in  (\ref{c1})-(\ref{c3}),(\ref{xxz1})-(\ref{xxz7}) 
and (\ref{f-path})-(\ref{c2-path}),(\ref{x1})-(\ref{xx2}) using GLPK  and CPLEX.

\subsection{Creating Instances}
Since there are no benchmark data sets for the {\sc FIFO Stack-Up} problem 
we generated randomly instances by an algorithm, which allows to
give the following parameters.

\begin{itemize}
\item $p_{\max}$ an upper bound on the number of stack-up places needed to
process the generated sequences 
\item $k$ number of sequences
\item $m$ number of pallets
\item $r_{\min}$ and $r_{\max}$ the minimum and maximum number of
bins per pallet
\item $d$ maximum number of sequences on which the bins of each pallet can
be distributed
%distance of the first and last bin for each pallet within the sequences
\end{itemize}

The idea is to compute a bin solution $B=(b_1,\ldots,b_n)$ with respect to
the given parameters and to distribute the bins to the $k$ sequences such that
the relative order will be preserved, i.e. $b_i$ will be placed left of $b_j$
in some sequence, if $i<j$ in $B$. 
The algorithm is shown in Figure \ref{seqgen-al}.

\begin{figure}[ht]
\hrule
{\strut\footnotesize \bf Algorithm {\sc Random Instances}} 
\hrule
\begin{tabbing}
xxx\=xxx\=xxx\=xxx \= xxx \= xxx \= xxx \= xxx\= xxx \= xxx \= xxx \= xxx\= xxx\=\kill
\#open $:=0$; open $:=\emptyset$; \>\>\> \>\> \>\> \>\> \>\>$\blacktriangleright$ Description of Functions and Variables:\\
$avg:=(r_{\min}+r_{\max})/2$ \>\>\> \>\> \>\> \>\> \>\> $\blacktriangleright$ random$(\ell,u)$: int // choose some value in $[\ell,u]$ \\
for all $i:=1$ to $m$ step $2$ do \>\> \>\> \>\> \>\> \>\>\>  $\blacktriangleright$ ~~~~~~~~~ at random \\
\> $r:=$ random$(0,r_{\max}-r_{\min})$ \>\>\> \>\> \>\> \>\> \> $\blacktriangleright$ $no[1\ldots m]$: int // number bins for each pallet \\
\> no[$i]:=avg + r$ \>\> \>\> \>\> \>\> \>\>  $\blacktriangleright$ $n$: int // total number of bins, $n=\sum_{i=1}^{m}no[i]$\\
\> no[$i+1]:=avg - r$\>\>\> \>\> \>\> \>\> \> $\blacktriangleright$ $seq[1\ldots m][1\ldots d]$: int // for each individual   \\
for all $i\in[m]$ and $j\in[d]$ do \>\>\> \>\> \>\> \>\> \> \> $\blacktriangleright$ ~~~~~~~~~ pallet there are up to $d$ sequences, on  \\
\> seq$[i][j]:=$ random$(1,k)$  \>\>\> \>\> \>\> \>\> \>  $\blacktriangleright$ ~~~~~~~~~ which to  distribute the bins\\
for all $i\in[k]$ do $q_i:=(~)$   \\
unproc $:= [m]$; $i:=0$;  \>\>\> \>\> \>\> \>\> \>\>     \\
while $i<n$ do  \>\>\> \>\> \>\> \>\> \>\>   \\
\> $i=i+1$ \\
\> if \#$open$$=p$ \\
\> \>  $plt$ $:=$ choose some pallet of open at random\\
\> else \\
\> \>  $plt$ $:=$ choose some pallet of open $\cup$ unproc at random\\
\> if $plt$ $\in$ unproc \>\>\> \>\> \>\> \>\>\>$\vartriangleright$ first bin of chosen pallet\\
\>\>  \#$open:=$ \#$open$+1; open $:=$ open $\cup$ $\{plt\}$; unproc := unproc $-$ $\{plt\}$ \\
\>$r:=$ random$(1,d)$\\
\>$s:=$ seq$[plt][r]$    \\
\> append bin $b_i$ to sequence $q_s$ \\
\> no[$plt]:=$ no[$plt$]-1\\
\> if no[$plt]=0$ \>\>\>\> \>\>\>\>\>\>$\vartriangleright$ last bin of pallet $plt$\\
\>\>  \#$open:=$ \#$open$-1; open $:=$ open - $\{plt\}$;  
\end{tabbing}
\hrule
\caption{Construction of random instances for the {\sc FIFO Stack-Up} problem.}
\label{seqgen-al}
\end{figure}

\subsection{Implementations}
We consider the breadth first search based solution
{\sc Find Optimal Pallet Solution} for the {\sc FIFO Stack-Up} problem 
given in Figure \ref{fig:algorithm4}. 
For several practical instance sizes the running time of this solution is 
too huge for computations, e.g. $m=254$ and $k=10$ lead a time of
$\bigo(n^2 \cdot (m+2)^k)=\bigo(n^2 \cdot 2^{80})$.
Therefore we used a cutting technique on the decision graph $G=(V,A)$ 
by restricting to vertices $v\in V$ representing configurations $C_Q(v)$ 
such that  $\#\OPEN(C_Q(v))\leq p_{\max}$ and increasing the value of
$p_{\max}$ by 5 until a solution is found, see Figure \ref{seqgen-al-xx}.
We have implemented algorithm {\sc Find Optimal Pallet Solution} as a single-threaded program 
in C++ on a standard Linux PC 
with 3.30 GHz CPU and 7.7 GiB RAM.

\begin{figure}[ht]
\hrule
{\strut\footnotesize \bf Algorithm {\sc Iterative process}} 
\hrule
\begin{tabbing}
xxx \= xxx \= xxx \= xxx \= xxx \= xxx \= xxx\= xxx \= xxx \= xxx\= xxx\= xxx \= xxx \= xxx   \kill
$p_{\max} := 5$;\\
level $:=0$; \> \>\>\>\>\> \>\>  $\vartriangleright$ last level without cuts\\
%put initial configuration to list $L[0]$ \\
$L[0] := (~(0,\ldots,0)~)$ \> \>\>\>\>\> \>\> $\vartriangleright$ $L[0]$: level 0 only contains the initial configuration\\
$pred[~(0,\ldots,0)~] := \emptyset$  \\
$val[~(0,\ldots,0)~] := 0$ \\
$open[~(0,\ldots,0)~]:=\emptyset$ \\
while BFS($p_{\max}$) yields no solution  do\\
\> $p_{\max} := p_{\max} + 5$
\end{tabbing}
\hrule
\caption{Iterative process of our cutting technique.}
\label{seqgen-al-xx}
\end{figure}

\begin{figure}[ht]
\hrule
{\strut\footnotesize \bf Algorithm {\sc BFS}($p$)} 
\hrule
\begin{tabbing}
xxx \= xxx \= xxx \= xxx \= xxx \= xxx \= xxx\= xxx \= xxx \= xxx\= xxx\= xxx \= xxx \= xxx   \kill
for $\ell :=$ level $+1$ to $m$ do \> \>\>\>\>\> \>\>\>\>\>\>\>  $\vartriangleright$  $\ell$: level in the decision graph\\
\>for each configuration $C \in L[\ell-1]$  do   \>\>\>\>\> \>\>\>\>\>\>\>  $\vartriangleright$ $L[\ell]$: list of nodes in level $\ell$ \\
\>\> if ($pred[C]\neq \emptyset$) \\
\>\>\> {\sc Extend Path}($C$) \\
\>\>  for $j := 1$ to $k$ do \\
\>\>\>    $C_s := (i_1, \ldots, i_j+1, \ldots, i_k)$ \\
\>\>\>   $O:= \OPEN[C]\cup \PL(b_{i_j+1})$ \\
\>\>\>  for $i := 1$ to $k$ do \>\>\>\>\>\>\>\>\>\> $\vartriangleright$ perform automatic transformation steps\\
\>\>\>\>  while the first bin in $q_i$ is destinated for an open pallet in $O$ \\
\>\>\>\>\>  let $C_s$ be the configuration obtained by removing the first bin of $q_i$ \\
\>\>\>  if ($C_s$ is not in $L[\ell]$ and $\# \OPEN(C_s) \leq p$) \>\> \> \>\>\>\>\>\>\> $\vartriangleright$ Cut\\ 
\>\>\>\>  $append(C_s,L[\ell])$ \\
\>\>\> $append(C,pred[C_s])$ \\
\>\>\>if ($C_s$ is final) return $C_s$ \\
\>\>  if $\ell-1 \neq p$  then $erase(C, L[\ell-1])$  \>\> \> \>\>\>\>\>\>\> \> $\vartriangleright$ remove $C$ from list $L[\ell-1]$\\
level$:=$ $p$; 
\end{tabbing}
\hrule
\caption{BFS.}
\label{seqgen-al-x}
\end{figure}

Our two linear programming approaches given in  (\ref{c1})-(\ref{c3}),(\ref{xxz1})-(\ref{xxz7}) 
and (\ref{f-path})-(\ref{c2-path}),(\ref{x1})-(\ref{xx2}) 
have been realized in GLPK v4.43 and CPLEX 12.6.0.0 and have been run
on the same machine. GLPK is single-threaded, while CPLEX uses all 4 cores
of the CPU.

\subsection{Evaluation}
First we consider our implementation of algorithm {\sc Find Optimal Pallet Solution}.
In Table \ref{table1} we list  our chosen parameters. 
For each assignment we randomly generated and solved 10 instances to compute  
the average time for solving an instance with the given parameters.
Our results show that we can solve practical instances on several thousand bins 
in a few minutes.

\begin{table}[ht]
$$
\begin{tabular}{|c|rrrrrrr|r|}
\hline
&\multicolumn{7}{|c|}{Instance} & BFS and cutting\\
   &  &&& &&&&  Pallet \\
 & $n$  & $p_{\max}$ & ~~~$m$ & ~~~$k$ &  ~$r_{\min}$ & ~$r_{\max}$ & ~~~$d$ &  Solution\\
\hline
1. & 1500 &  14 & 100 & 8 & 10 & 20 & 4 & 0.14 \\
2. & 1500 & 14 & 100 & 8 & 10 & 20 & 6 & 0.12 \\ 
3. & 1500 & 14 & 100 & 8 & 10 & 20 & 8 & 0.08 \\ 

4. & 2000 & 14 & 100 & 8 & 10 & 30 & 4 & 0.17 \\ 
5. & 2000 & 14 & 100 & 8 & 10 & 30 & 6 & 0.15 \\ 
6. & 2000 & 14 & 100 & 8 & 10 & 30 & 8 & 0.13 \\ 

7. & 2500&14 & 100 & 8 & 10 & 40 & 4 & 0.22 \\ 
8. & 2500 &14 & 100 & 8 & 10 & 40 & 6 & 0.10\\ 
9. & 2500 &14 & 100 & 8 & 10 & 40 & 8 & 0.07\\ 

\hline
10.& 6000 &18 & 300 & 10 & 15 & 25 & 5 & 5.78  \\ 
11.& 6000 &18 & 300 & 10 & 15 & 25 & 7 & 3.72 \\ 
12.& 6000 &18 & 300 & 10 & 15 & 25 & 10 & 2.01 \\ 

13.& 7500 &18 & 300 & 10 & 15 & 35 & 5 &  7.78 \\ 
14.& 7500 &18 & 300 & 10 & 15 & 35 & 7 &  2.83 \\ 
15.& 7500 &18 & 300 & 10 & 15 & 35 & 10 & 2.07 \\ 

16. & 9000& 18 & 300 & 10 & 15 & 45 & 5 &  4.21 \\ 
17. & 9000& 18 & 300 & 10 & 15 & 45 & 7 &  2.23\\ 
18. & 9000& 18 & 300 & 10 & 15 & 45 & 10 & 1.45 \\ 

\hline
19. & 12500 &22 & 500 & 12 & 20 & 30 & 6 & 92.52 \\ 
20. & 12500 &22 & 500 & 12 & 20 & 30 & 9 & 52.74\\ 
21. & 12500 &22 & 500 & 12 & 20 & 30 & 12 & 42.81 \\

22. & 15000 &22 & 500 & 12 & 20 & 40 & 6 &  103.24\\ 
23. & 15000 &22 & 500 & 12 & 20 & 40 & 9 &  49.54\\
24. & 15000 &22 & 500 & 12 & 20 & 40 & 12 & 32.85 \\ 

25. &17500 &22 & 500 & 12 & 20 & 50 & 6 & 88.52  \\ 
26. &17500 &22 & 500 & 12 & 20 & 50 & 9 &  38.61\\ 
27. &17500 &22 & 500 & 12 & 20 & 50 & 12 & 41.70
 \\ 
\hline
\end{tabular}
$$
\caption{Running times in seconds for randomly generated
instances for finding optimal solutions of 
the {\sc FIFO Stack-Up} problem by algorithm {\sc Find Optimal Pallet Solution}.  \label{table1}}
\end{table}

Next we consider our two linear programming models realized in GLPK and CPLEX.
Since the size of the instances of Table \ref{table1} was so high that none of the ILP 
approaches was able to solve them, we generated much smaller parameters in the instances 
of Table \ref{table-ex}.
For each assignment we randomly generated and solved 10 instances 
in order to compute the average time for solving the same instances with 
the given parameters by our two linear programming models using GLPK and CPLEX.

Our results show that algorithm {\sc Find Optimal Pallet Solution} can be used to solve practical
instances on several thousand bins of the {\sc FIFO Stack-Up} problem. Our two linear programming
approaches can only be used to handle instances up to 100 bins and less than 10 pallets.
Since it is a commercial product of high licence cost CPLEX can solve the instances much faster 
than the open-scource solver GLPK and, 
as expected, the pallet solution approach is much better than the bin solution approach.

\begin{table}
\begin{center}
\begin{tabular}{|c|rrrrrrr|r|r|r|r|}
\hline
&\multicolumn{7}{|c|}{Instance} & \multicolumn{2}{c|}{GLPK} & \multicolumn{2}{c|}{CPLEX} \\
&  &  &    &&& & & Bin  &Pallet  & Bin  & Pallet    \\
&$n$ & $p_{\max}$ & $m$  & $k$  &$r_{\min}$ &$r_{\max}$& $d$ &  Solution  & Solution & Solution &  Solution   \\

\hline
1.&15  &  2 & 3  & 2 & 4 & 6 & 2 & 41.4 & 0.1   & 0.1 & 0.1  \\ %  -P 2 -M 3 -K 2 -R 4 6 -D 4 f15
2.&20  &  2 & 4  & 2 & 4 &6  & 2 &-     & 0.1   & 1.7 & 0.1                \\   % -P 2 -M 4 -K 2 -R 4 6 -D 4 f14
3. &30  &  4 & 5  & 4 & 4 &8  & 2 &-     & 0.7   & 66.8 &  0.2\\    % -P 4 -M 5 -K 4 -R 4 8 -D 4 f13
4.&48  &  4 & 6  & 4 & 6 &10 & 2 &-     & 2.5   & 932.6 & 1.2\\ % -P 4 -M 6 -K 4 -R 5 10 -D 4 f12
5. & 64  &  4 & 8  & 5 & 6 &10 & 2 &-     & 684.6 & -  & 16.3                   \\  %-P 4 -M 8 -K 5 -R 5 10 -D 4 f11
6. &100 &  5 & 10 & 5 & 5 &15 & 2 &-     & -     & -  & 282.3          \\  %-P 5 -M 10 -K 5 -R 5 15 -D 4 f10
\hline
\end{tabular} 
\end{center} 
\caption{Running times in seconds for randomly generated
instances for finding optimal solutions of the FIFO stack-up problem. Running times of more than 1800 seconds = 30 minutes are 
indicated by a bar (-).\label{table-ex}}
\end{table}

%%%%%%%%%%%%%%%%%%%%%%%%%%%%%%%%%%%%%%%%%%%%%%%%%%%%%%%%%%%%%%%%%%%%%%%%%%%%%%%%
\section{Conclusions and Outlook}
%%%%%%%%%%%%%%%%%%%%%%%%%%%%%%%%%%%%%%%%%%%%%%%%%%%%%%%%%%%%%%%%%%%%%%%%%%%%%%%%

In this paper we consider three graph models for the {\sc FIFO Stack-Up} problem.
Based on these models we have shown a breadth first search solution 
and two linear programming solutions to solve the problem. 
Further we have given parameterized algorithms w.r.t.\ several
parameters which are summarized in Table \ref{table-fpt}.
We also could give a first approximation result for minimizing the number of 
stack-up places.

\begin{table}[h]
$$
\begin{array}{|l|ccccc|}
\hline
\text{parameter} & k & p & m &n & (k,m)  \\ 
\hline
\fpt   & ? & ? & + & + & + \\
\xp    & + & +  & +  & +  & + \\
\hline
\end{array}
$$
\caption{Parameterized complexity of the {\sc FIFO Stack-Up} problem}\label{table-fpt}
\end{table}

In our future work we want to determine 
the complexity of the  {\sc FIFO Stack-Up} problem
for $d_Q\in\{2,\ldots,5\}$.
Further we intend to find better
approximation algorithms, try to improve the running time of
the given parameterized algorithms, and
explore the existence of fpt-algorithms w.r.t.\ parameters $k$ and $p$. 
Due Tamaki (Section 6 in \cite{Tam11}) the existence of an
fpt-algorithm for the directed pathwidth problem w.r.t. the standard parameter is still open. By Theorem \ref{th-pw-gbefore} 
such an algorithm would imply an fpt-algorithm for the {\sc FIFO Stack-Up} problem w.r.t. $p$, and vice versa.

We are also interested in on-line algorithms for 
instances where we only know the first $c$ bins
of every sequence instead of the complete sequences \cite{Bor98,FW98}.
Especially, we are interested in the answer to the following
question: Is there a $d$-competitive on-line algorithm?
Such an algorithm must compute a processing
of some $Q$ with at most $p \cdot d$ stack-up places, if $Q$ can be processed
with at most $p$ stack-up places. First approaches for on-line algorithms
for controlling palletizers are presented in \cite{GRW16c}.

In real life the bins arrive at the stack-up system on the main conveyor
of a pick-to-belt orderpicking system. That means, the distribution of bins
to the sequences, for example by some pre-placed cyclic storage conveyor, 
has to be computed. Up to now we consider the distribution
as given. We intend to consider how to compute
an optimal distribution of the bins from the main conveyor onto the sequences
such that a minimum number of stack-up places is
necessary to stack-up all bins from the sequences.

%%%%%%%%%%%%%%%%%%%%%%%%%%%%%%%%%%%%%%%%%%%%%%%%%%%%%%%%%%%%%%%%%%%%%%%%%%

%\bibliographystyle{alpha}
\bibliographystyle{plain}
%\bibliography{../bib}
\bibliography{/home/gurski/bib.bib}
%%%%%%%%%%%%%%%%%%%%%%%%%%%%%%%%%%%%%%%%%%%%%%%%%%%%%%%%%%%%%%%%%%%%%%%%%%

\end{document}